\documentclass[sigconf,screen,9pt]{acmart}
\usepackage{booktabs}
\usepackage{multirow}
\usepackage{textcomp}
\usepackage{bm}
\usepackage{makecell}
\usepackage{cancel}
\usepackage{subcaption}
\usepackage[ruled,linesnumbered,noresetcount,vlined]{algorithm2e}

\newcommand{\mycommentstyle}[1]{\color{black}{\small #1}}
\SetKwComment{AlgComment}{\mycommentstyle{// }}{}

\usepackage{color}   %导入 color，xcolor 宏包均可

\definecolor{b}{rgb}{0,0,1}

\makeatletter
\newcommand*{\bluestart}{\@ifnextchar\bgroup{\bluestart@}{\color{blue}}}
\newcommand*{\bluestart@}[1]{{\textcolor{blue}{#1}}}
\makeatother

\usepackage{cleveref}

\newcommand{\ID}{\mathrm{ID}}

\newcommand{\OS}{\mathrm{OS}}

\newtheorem{theorem}{Theorem}[section]
\newtheorem{lemma}{Lemma}[section]
\newtheorem{assumption}{Assumption}[section]
\newtheorem{remark}{Remark}[section]
\AtBeginDocument{%
	}

\copyrightyear{2026}
\acmYear{2026}
\setcopyright{cc}
\setcctype{by}
\acmConference[E-Energy '26]{The 17th ACM International Conference on Future and Sustainable Energy Systems}{June 22--25, 2026}{Banff, AB, Canada}
\acmBooktitle{The 17th ACM International Conference on Future and Sustainable Energy Systems (E-Energy '26), June 22--25, 2026, Banff, AB, Canada}
\acmDOI{10.1145/3744255.3811718}
\acmISBN{979-8-4007-2011-6/2026/06}

\begin{document}
	
	\title{Unsupervised Learning for AC Optimal Power Flow with Fast Physics-Aware Layer}
	
	
	\author{Jiebao Zhang$^\dagger$}
	\email{zhangjb2023@shanghaitech.edu.cn}
	\affiliation{
	\institution{ShanghaiTech University}
	\city{Shanghai}
	\country{China}
}

	\author{Haoyu Yan$^\dagger$}
	\email{yanhy2023@shanghaitech.edu.cn}
	\affiliation{%
		\institution{ShanghaiTech University}
		\city{Shanghai}
		\country{China}
	}
	
	\thanks{$^\dagger$The two authors contribute equally to this work}
	\thanks{$^*$Corresponding author}
	
	\author{Zhichao Sheng}
	\email{zcsheng@shu.edu.cn}
%	\email{{zcsheng, hw\_yu}@shu.edu.cn}
	\affiliation{%
		\institution{Shanghai University}
		\city{Shanghai}
		\country{China}
	}
	
	\author{Hongwen Yu}
	%	\email{hw\_yu@shu.edu.cn}
	%	\author{Hongwen Yu}
	\email{ hw\_yu@shu.edu.cn}
	\affiliation{%
		\institution{Shanghai University}
		\city{Shanghai}
		\country{China}
	}
	
	\author{Shuang Ye}
	\email{yes@sari.ac.cn}
	\affiliation{%
		\institution{ShanghaiTech University}
		\city{Shanghai}
		\country{China}
	}
	
	\author{Haoyu Wang}
	\email{wanghy@shanghaitech.edu.cn}
	\affiliation{%
		\institution{ShanghaiTech University}
		\city{Shanghai}
		\country{China}
	}
	
	\author{Ye Shi$^*$}
	\email{shiye@shanghaitech.edu.cn}

	\affiliation{%
	\institution{ShanghaiTech University}
	\city{Shanghai}
	\country{China}
}

	\renewcommand{\shortauthors}{Jiebao Zhang et al.}
	
	\begin{abstract}
		Learning to solve the Alternating Current Optimal Power Flow (AC-OPF) problem by neural networks (NNs) is a promising approach in real-time applications.
		Existing methods to ensure the physical feasibility of NN outputs embed a power flow (PF) solver within networks.
		However, the gradient through the PF solver, namely, implicit differentiation, needs manual Jacobian derivation and the solution of linear systems, which is computationally prohibitive and hinders integration with modern automatic differentiation (AD) frameworks.
		To address these challenges, we propose FPL-OPF, a novel unsupervised learning framework that incorporates a Fast Physics-aware Layer for AC-OPF problems. FPL-OPF embeds a fast PF iterative solver within the NN and takes solely the last few or even the final iterations into the AD graph. This design ensures high computational efficiency for both the forward and backward passes, circumventing complex custom backward implementations.
		Theoretically, we rigorously prove that the gradient from this design serves as a high-fidelity surrogate of the true implicit gradient under mild conditions.
		Extensive experiments demonstrate that FPL-OPF achieves significant speedups over state-of-the-art unsupervised learning approaches, while maintaining near-zero constraint violations and competitive optimality.
		Our code is available at \url{https://github.com/wowotou1998/fpl-opf}
%		 \footnote{Our code is available at \url{https://github.com/wowotou1998/fpl-opf}}.
	\end{abstract}
	
%	\begin{CCSXML}
%		<ccs2012>
%		<concept>
%		<concept_id>10010147.10010257.10010293.10010294</concept_id>
%		<concept_desc>Computing methodologies~Neural networks</concept_desc>
%		<concept_significance>300</concept_significance>
%		</concept>
%		</ccs2012>
%	\end{CCSXML}
	
	
	%\ccsdesc[500]{Do Not Use This Code~Generate the Correct Terms for Your Paper}
	%\ccsdesc[300]{Do Not Use This Code~Generate the Correct Terms for Your Paper}
	%\ccsdesc{Do Not Use This Code~Generate the Correct Terms for Your Paper}
	%\ccsdesc[100]{Do Not Use This Code~Generate the Correct Terms for Your Paper}
	
	
	

	\maketitle
	
	\section{Introduction}

	\begin{table*}[ht]
		\centering
		\caption{Comparison of learning-based AC-OPF methods with feasibility enforcement .}
		\label{tab:related_work_table}
		\resizebox{0.95\linewidth}{!}{ 
			\begin{tabular}{lccccc}
				\hline
				\textbf{Method}                                                                           & \textbf{Constraint Satisfaction}    & \textbf{Backward Complexity} & \textbf{Implementation Effort}    & \textbf{Training Efficiency} & \textbf{Inference Speed} \\ \hline
				Supervised NN~\cite{fiorettoPredictingACOptimal2020,huangDeepOPFVSolvingACOPF2022,panDeepOPFFeasibilityOptimizedDeep2023}                                      & Limited (data-dependent)            & No Backward                  & Low                               & High                         & High                     \\
				Penalty-based losses~\cite{huang2021deepopf, huang2024unsupervised, chen2022unsupervised,owerko2024unsupervised, yang2024topology}                       & Soft (violation possible)           & No Backward                  & Low                               & High                         & High                     \\
				Post-processing methods~\cite{panDeepOPFDeepNeural2019,panDeepOPFDeepNeural2021,zamzam2020learning, huangDeepOPFVSolvingACOPF2022} & Moderate (post-projection)          & No Backward                  & Moderate                          & Moderate                     & Moderate                 \\
				DC3~\cite{donti2021dc3}, DeepLDE~\cite{kim2025DeepLDE}, \cite{chen2025PhysicsInformedGradientEstimationAcceleratingDeep}                                                                  & Strong (exact PF embedding)         & High (large Jacobian system) & High (custom backward)            & Low                          & Moderate                 \\
				FRMNet~\cite{hanFRMNetFeasibilityRestoration2024}                                         & Strong (feasibility restoration)    & High (large Jacobian system) & High (custom backward)            & Moderate                     & Moderate                 \\
				OptNet Embedded~\cite{jia2024OptNetEmbeddedDataDrivenApproachOptimalPower}                & Strong (lifted constraints)         & High (large Jacobian system) & High (custom backward)            & Low                          & Moderate                 \\
				QCQP-Net~\cite{zeng2024QCQPNetReliablylearningfeasiblealternating}                        & Strong (relaxed PF equations)       & High (large Jacobian system) & High (custom backward)            & Low                          & Moderate                 \\ \hline
				\textbf{FPL-OPF (ours)}                                                                  & \textbf{Strong (fixed-point layer)} & \textbf{Low (ID approximation)}   & \textbf{Low (no custom backward)} & \textbf{High}                & \textbf{High}            \\ \hline
			\end{tabular}
		}
	\end{table*}
	The Alternating Current Optimal Power Flow (AC-OPF) problem is a cornerstone of modern power system operations, tasked with determining the most economic and secure generator dispatch. 
	The AC-OPF problem is inherently non-convex and NP-hard, making global optimization extremely challenging~\cite{shi2017global, shi2018global}.
	With the increasing penetration of variable renewable energy sources and the growing complexity of the grid, it becomes critically important to solve the non-convex, non-linear AC-OPF problem in real-time. 
	Traditional numerical solvers, while accurate, often struggle to meet the severe time requirements of real-time market operations, which can be as frequent as every five minutes~\cite{diehl2019warm}. This computational bottleneck motivates the development of faster and more efficient methods.
	To this end, the {learning-based} paradigm has emerged as a promising alternative, leveraging deep learning to approximate AC-OPF solutions~\cite{huangDeepOPFVSolvingACOPF2022,panDeepOPFFeasibilityOptimizedDeep2023,zhouDeepOPFFTOneDeep2023}. 
	Supervised learning approaches, which train neural networks (NNs) on large datasets of pre-solved AC-OPF instances, have demonstrated the ability to produce solutions rapidly~\cite{fiorettoPredictingACOptimal2020,huangDeepOPFVSolvingACOPF2022,panDeepOPFFeasibilityOptimizedDeep2023,singh2021learning}. 
	However, their performance is fundamentally limited by the quality and diversity of the training data, which is computationally expensive to generate and may contain sub-optimal solutions due to the problem's non-convexity. 
	Unsupervised learning methods circumvent this data-dependency by directly incorporating the physics of the optimization problem, i.e., the objective and constraints, into the loss function.
	\par 
	A fundamental challenge remains in guaranteeing that the NN outputs strictly satisfy the physical and operational constraints of the power grid.
	Existing approaches address this feasibility issue with some limitations. 
	Penalty-based approaches~\cite{huang2021deepopf, huang2024unsupervised, chen2022unsupervised,owerko2024unsupervised, yang2024topology} integrate constraint violations as soft penalties within the loss function. While straightforward, they inherently fail to guarantee strict feasibility and are highly sensitive to hyperparameter tuning~\cite{huang2024unsupervised}.
	Post-processing methods~\cite{panDeepOPFDeepNeural2019,panDeepOPFDeepNeural2021,zamzam2020learning, huangDeepOPFVSolvingACOPF2022} attempt to project infeasible predictions onto the feasible regions. However, for the non-convex AC-OPF, this projection may introduce significant computational overhead and is detached from the end-to-end training pipeline, limiting the model's ability to learn physics-aware outputs.
	More recently, implicit differentiation layers~\cite{donti2021dc3,chen2025PhysicsInformedGradientEstimationAcceleratingDeep,kim2025DeepLDE,jia2024OptNetEmbeddedDataDrivenApproachOptimalPower,hanFRMNetFeasibilityRestoration2024,zeng2024QCQPNetReliablylearningfeasiblealternating} have been proposed to enforce strict physical constraints by embedding power flow (PF) solvers directly into networks.
	However, these implicit layer-based approaches, while powerful, face some significant practical limitations.
	Even though specialized PF solvers are efficient in the forward process, e.g., fast decoupled power flow (FDPF)~\cite{stott1974FastDecoupledLoadFlow,chen2022unsupervised,chen2025PhysicsInformedGradientEstimationAcceleratingDeep}, the exact implicit gradient is computationally burdensome, as it necessitates solving a large-scale linear system involving the constraint Jacobian at every backward pass~\cite{chen2025PhysicsInformedGradientEstimationAcceleratingDeep}. Second, the implementation is often complex and error-prone, requiring hand-crafted derivations of the Jacobian matrix. This manual process runs contrary to the streamlined workflow of modern automatic differentiation (AD) frameworks, hindering development and broader adoption. 
	\par 
	Fortunately, recent advances in implicit differentiation have shown that approximations can effectively guide network training. Methods such as truncated backpropagation~\cite{shaban2019truncated}, Jacobian-Free~\cite{fungJFBJacobianFreeBackpropagation2022}, phantom gradients~\cite{gengTrainingImplicitModels2022}, and one-step differentiation~\cite{bolteOneStepDifferentiationIterative2023} offer high-fidelity gradients at a fraction of the computational cost and integrate seamlessly into standard AD pipelines, presenting a compelling path forward.
	\par 
	\textbf{Contributions.}
	These challenges motivate our proposal of FPL-OPF, a novel unsupervised learning framework that incorporates a Fast Physics-aware Layer for fast solving AC-OPF problems.
	Our framework is architected around two core design components. 
	For the forward pass, we model the PF balance as a fixed-point problem and embed an efficient FDPF solver as an implicit layer. 
	Unlike standard Newton-Raphson (NR) methods with $\mathcal{O}(n^3)$ complexity per iteration, the FDPF solver achieves a more favorable $\mathcal{O}(n^2)$ complexity, where $n$ denotes the number of buses in the power grid. 
	For the backward pass, we harness standard AD solely for the final, strongly contractive refinement stage, comprising a single NR step or a few FDPF iterations, which circumvents complex custom backward implementations.
	We provide rigorous theoretical justification for this manner, proving that under near-convergence conditions, the approximation implicit gradient calculated by modern AD frameworks maintains a positive cosine similarity with the true one, ensuring it is a high-fidelity descent direction.
	This elegant backward design yields $\mathcal{O}(n^2)$ complexity and leads to a drastic reduction in computational overhead, resulting in a highly efficient end-to-end framework. 
%	Furthermore, our design circumvents complex custom backward implementations by integrating FDPF iterations directly into the computational graph and using AD for only the refinement stage. 
	Extensive experiments on diverse power grids show that FPL-OPF achieves orders-of-magnitude faster inference than the oracle solver and substantially faster to train than state-of-the-art unsupervised approaches, while maintaining near-zero constraint violations and optimality gaps below 3\%. These results validate FPL-OPF as an efficient framework for real-time, non-convex AC-OPF applications.
	\par
	The main contributions of this work are threefold:
	\begin{itemize}
		\item We introduce FPL-OPF, a fast physics-aware and unsupervised learning framework for solving the AC-OPF problem. By embedding an efficient FDPF iterative solver as a fixed-point implicit layer, FPL-OPF achieves $\mathcal{O}(n^2)$ complexity for both the forward and backward passes, offering a substantial reduction compared with existing implicit-layer-based approaches.
		\item We theoretically establish that the approximate implicit gradient obtained from the refinement stage is well aligned with the true gradient under near-convergence conditions, thereby guaranteeing a high-fidelity descent direction for an efficient training pipeline.
		%		 and providing a new theoretical foundation for one-step implicit differentiation in power flow networks.
		\item We empirically validate FPL-OPF on diverse power grid benchmarks, showing that it achieves substantially faster training than state-of-the-art unsupervised approaches and orders-of-magnitude faster inference than the oracle solver, while maintaining near-zero constraint violations and optimality gaps below 3\%.
	\end{itemize}
	\par 
	The rest of the paper is organized as follows. 
	\Cref{sec:related_work} reviews the related work on learning-based AC-OPF with feasibility considerations.
	\Cref{sec:acopf_formulation} introduces the AC-OPF problem formulation.
	\Cref{sec:methodology} details the proposed FPL-OPF framework. \Cref{sec:experiments} reports experimental evaluations on diverse power grids. Finally, \Cref{sec:conclusion} concludes the paper and outlines future research directions.
	
	\section{Related Work}
	\label{sec:related_work}
	Many works, as summarized in \Cref{tab:related_work_table}, have been proposed to enforce the network outputs' feasibility that satisfy the strict physical and operational constraints of power systems, including penalty-based losses, post-processing methods, and implicit differentiation layers.
	\par 
	Penalty-based methods represent a straightforward approach that formulates the violations of constraints, such as the PF balance equations, generator limits, and line thermal ratings, as penalty terms inside the loss function. 
	Kejun et al.~\cite{chen2022unsupervised} utilize a modified augmented Lagrangian function as the training loss, and the multipliers are adjusted dynamically based on the degree of constraint violation.
	DeepOPF-NGT~\cite{huang2021deepopf,huang2024unsupervised} proposes a properly designed loss function and develops an adaptive learning rate algorithm to dynamically balance the gradient contributions from different loss terms during training.
	PG-GNN~\cite{yang2024topology} trains a graph neural network, which takes the Karush-Kuhn–Tucker condition on the Lagrangian function as the training goal.
	However, this approach essentially does not guarantee the feasibility of the network's output. Furthermore, penalty-based methods are highly sensitive to the penalty coefficients~\cite{huang2024unsupervised}.
	\par
	Post-processing methods project the infeasible prediction from the NN onto the feasible region. 
	DeepOPF~\cite{panDeepOPFDeepNeural2019} proposes a quadratic programming projection that finds the closest feasible solution within the polyhedron defined by DC-OPF linear constraints. For the security-constrained DC-OPF problem, Pan et al. ~\cite{panDeepOPFDeepNeural2021} develop an $L_{1}$-projection procedure that recovers feasibility by solving a linear programming problem. In the AC-OPF learning task, Zamzam et al.~\cite{zamzam2020learning} ensure solution feasibility by solving PF equations and adjusting reactive power outputs through a modified PF problem when limits are violated. DeepOPF-V~\cite{huangDeepOPFVSolvingACOPF2022} predicts voltages of all buses and uses them to reconstruct the remaining variables without solving non-linear AC PF equations with a fast post-processing process to enforce the box constraints.
	While effective, the projection itself may introduce significant computational overhead.
	Critically, the post-processing procedure is detached from the end-to-end training pipeline, fundamentally limiting the network's ability to produce inherently physics-aware outputs.
	\par 
	%		A line of work follows the seminal study.
	Implicit differentiation layers enforce the constraint satisfaction via embedding physical constraints directly within the network architecture, which fundamentally preserve the physical characteristics of the output of NNs, \textit{i.e.},  Kirchhoff’s Current Law inside the nonlinear PF equations. By defining a layer as the solution to the PF equations, feasibility can be more rigorously maintained.
	DC3~\cite{donti2021dc3}, a deep learning method for optimization problems with hard constraints that implicitly completes partial solutions to satisfy PF equality constraints and applies differentiable gradient-based corrections to enforce inequality constraints, yielding near-optimal and feasible solutions in AC-OPF applications.
	Similarly, DeepLDE~\cite{kim2025DeepLDE} adopts an equation embedding strategy to strictly satisfy equality constraints by solving the PF equations. Distinct from DC3, DeepLDE enforces inequality constraints through a primal-dual learning scheme rather than applying iterative corrections.
	The work in~\cite{chen2025PhysicsInformedGradientEstimationAcceleratingDeep}  proposes a novel physics-informed gradient estimation that simplifies the computationally intensive backpropagation.
	This is achieved by utilizing a batch-mean approximation of the Jacobian tensor, linearized and decoupled Jacobian models, and a reduced set of branch flow constraints that are most likely to be active.
	FRMNet~\cite{hanFRMNetFeasibilityRestoration2024} guarantees feasibility through a differentiable Feasibility Restoration Mapping block. This optimization block maps the potentially infeasible output from the neural network to the closest point within the feasible region, thereby satisfying the PF equations and other operational constraints.
	OptNet-embedded data-driven approach~\cite{jia2024OptNetEmbeddedDataDrivenApproachOptimalPower} utilizes OptNet layer~\cite{OptNet} as a proxy for the AC-OPF problem. Feasibility is enforced by first lifting the problem to a higher-dimensional space where PF equations are approximated as linear constraints, which are then explicitly included in the embedded optimization layer. 
	QCQP-Net (Quadratically Constrained Quadratic Program)~\cite{zeng2024QCQPNetReliablylearningfeasiblealternating}, which addresses potential infeasibility by solving a relaxed version of the PF equations when an exact solution may not exist for a given NN prediction. This relaxation is formulated as a differentiable non-convex QCQP that minimizes constraint violations, ensuring a stable solution can always be found during training.
	Despite their improved feasibility guarantees, implicit differentiation layers often introduce substantial computational burden during training. 
%	In particular, the backward pass requires differentiating through the embedded physical or optimization layer, which typically involves solving large-scale linear systems associated with the Jacobian of the PF equations or the KKT system of the embedded optimization problem. 
	\section{AC-OPF Problem Formulation}
	\label{sec:acopf_formulation}
	The AC-OPF is a constrained optimization problem aiming to find a power system operating point that minimizes a certain objective function, typically the total active power generation cost, while satisfying both physical laws and operational limits. Mathematically, the problem can be formulated as \eqref{eq:acopf}:
	\begin{subequations}
		\begin{align}
			&\min_{\bm{V}, \bm{\theta}, \bm{P}^{g}, \bm{Q}^{g}} \; \sum_{i} c_{i}(P^{g}_{i}) \label{eq:objective}\\
			\text{s.t.} \; 
			& P^{g}_{i} - P^{d}_{i} = \sum_{(i,j) \in \mathcal{L}} P^{l}_{ij}, \quad \forall i \in \mathcal{N}, \label{eq:pf_p} \\
			& Q^{g}_{i} - Q^{d}_{i} =\sum_{(i,j) \in \mathcal{L}} Q^{l}_{ij}, \quad \forall i \in \mathcal{N}, \label{eq:pf_q} \\ 
			& P_{ij}^l={G}_{ij}V_{i}^{2} - V_{i} V_{j} \left( {B}_{ij} {\sin}\theta_{ij}+{G}_{ij}{\cos}\theta_{ij} \right) \; \forall(i,j)\in\mathcal{L} \label{eq:pl_p} \\
			& Q_{ij}^l=-{B}_{ij}V_{i}^{2} - V_{i} V_{j} \left( {G}_{ij} {\sin}\theta_{ij}-{B}_{ij}{\cos}\theta_{ij} \right) \; \forall (i,j)\in\mathcal{L}  \label{eq:pl_q}\\
			&\underline{P^{g}_{i} }\le P^{g}_{i} \le \overline{P^{g}_{i}}, \quad  \underline{Q^{g}_{i}} \le Q^{g}_{i} \le \overline{Q^{g}_{i}}, \quad \forall i \in \tilde{\mathcal{G}}, \label{eq:gen_pq_bound}\\ 
			&(P_{ij}^{l})^{2} + (Q_{ij}^{l})^{2}  \leq \overline{S_{ij}}, \quad  \forall (i,j) \in \mathcal{L}, \label{eq:line_bound} \\
%			&  \forall i \in \mathcal{N}, \label{eq:v_bound} \\ 
			&\underline{V_{i}} \le V_i \le \overline{V_i}, \quad \underline{\theta_{ij}} \leq  \theta_{ij} \leq \overline{\theta_{ij}}, \quad \forall i \in \mathcal{N}, \forall j \in \mathcal{N}. \label{eq:v_theta_ineq}
		\end{align}
		\label{eq:acopf}
	\end{subequations}
	The objective function \eqref{eq:objective} minimizes the total active power generation cost, where $c_{i}(\cdot)$ denotes the cost function for generator $i$.
	The vector $\bm{y} = [ \bm{V}^\top, \bm{\theta}^\top, (\bm{P}^{g})^\top, (\bm{Q}^{g})^\top ]^{\top}$ represents the complete set of system decision variables.
	\par
	The equality constraints \eqref{eq:pf_p}--\eqref{eq:pl_q}, denoted by $\bm{h}(\bm{y}, \bm{d}) = \bm{0}$, enforce the nodal active and reactive power balance, where $\bm{d} = [  (\bm{P}^{g})^\top, (\bm{Q}^{g})^\top ]^{\top}$ represents the load profile, and $\mathcal{N}$ and $\mathcal{L}$ denote the sets of buses and transmission lines, respectively. The branch power flows $P_{ij}^l$ and $Q_{ij}^l$ are derived from the voltage magnitudes $\bm{V}$, voltage angles $\bm{\theta}$, and the conductance $G_{ij}$ and susceptance $B_{ij}$ of the admittance matrix. 
	\par
	The inequality constraints  \eqref{eq:gen_pq_bound}--\eqref{eq:v_theta_ineq}, denoted by $\bm{g}(\bm{y}) \le \bm{0}$, define the operational limits  within their respective lower ($\underline{\cdot}$) and upper ($\overline{\cdot}$) security bounds. 
	\eqref{eq:gen_pq_bound} constrains the active and reactive outputs for the set of generators $\tilde{\mathcal{G}}$ (within the slack bus) .
	\eqref{eq:line_bound} limits the apparent power flow on line $(i,j)$ to its thermal rating $\overline{S_{ij}}$.
	\eqref{eq:v_theta_ineq} keep voltage magnitudes and phase angle differences.
	\section{Methodology}
	\label{sec:methodology}
	In this section, we present the details of our proposed FPL-OPF framework. We begin by describing the overall architecture of our physics-aware learning model, which embeds the PF equations as a differentiable implicit layer.

	\subsection{FPL-OPF Unsupervised Learning Framework}
	
	\begin{algorithm}[!t]
		\caption{Training of FPL-OPF}\label{alg:primal_dual_training}
		%		\SetAlgoLined
		\KwIn{{Training dataset $\mathcal{S}=\left\{ \bm{d}^{l} \right\}_{l=1}^{n}$,
				outer iterations $\mathcal{T}$, inner iterations $\mathcal{I}$,
				learning rate $\eta_{\bm{\phi}}$, Lagrangian step size $\eta_{\bm{\lambda}}, \eta_{\bm{\nu}}$,
				initial Lagrange multipliers $\bm{\lambda}^0, \bm{\nu}^0$}
		}
		\KwOut{{The optimal NN weights $\bm{\phi}$}}
		\textbf{Initialize} \\
%		\Begin{
			\For{$k = 0,1,..., \mathcal{T}$}{
				\For{$i = 0,1,..., \mathcal{I}$}{
					\For{\textbf{all} $\bm{d} \in \mathcal{S}$}{
						$\bm{x} \gets \mathcal{M}_{\bm{\phi}}(\bm{d})$ \\
						$\bm{z} \gets$ Obtain $\bm{z}$ via Algorithm \ref{alg:hybrid_solver} \\
						$\tilde{\bm{z}} \gets$ Obtain $\tilde{\bm{z}}$ via Post-completion in \Cref{subsec:pf_solver} \\
						$\bm{y} \gets$ Obtain ${\bm{y}}$ via Concatenation of $\bm{x}, \bm{z}$, and $\tilde{\bm{z}}$\\
						$\nabla_{\bm{\phi}} \mathcal{L}(\bm{y},\bm{\lambda}, \bm{\nu}) \gets$ Calculated via \eqref{eq:grad_for_phi} \\
						$\bm{\phi} \gets \bm{\phi} - \eta_{\phi} \nabla_{\bm{\phi}} \mathcal{L}(\bm{y},\bm{\lambda}, \bm{\nu})$
					}
				}
				$\bm{\lambda} \gets \bm{\lambda} +  \eta_{\bm{\lambda}} \sum_{\bm{d} \in \mathcal{S}} \max(\bm{0}, \bm{g}(\bm{y},\bm{d})) $ \\
				$\bm{\nu} \gets \bm{\nu} +  \eta_{\bm{\nu}} \sum_{\bm{d} \in \mathcal{S}} |\bm{h}(\bm{y},\bm{d})| $ \\
%			}
			\Return{$\bm{\phi}$}
		}
	\end{algorithm}
	
	For the parametric version of problem \eqref{eq:acopf}:
	\begin{align*}
		\underset{\bm{y}}{\arg \min} \; f(\bm{y},\bm{d}) \quad \text{s.t. } \bm{h}(\bm{y},\bm{d}) = \bm{0}, \; \bm{g}(\bm{y},\bm{d}) \leq 0.
	\end{align*}
	Our framework trains a deep neural network $\mathcal{M}_{\bm{\phi}}$ parameterized by weights $\bm{\phi}$, to rapidly predict high-quality solutions to the AC-OPF problem in an unsupervised, physics-aware manner. 
	\par
	We formulate the training $\mathcal{M}_{\bm{\phi}}$ as a min-max game on the Lagrangian~\cite{park2023SelfSupervisedPrimalDualLearningConstrainedOptimization,fioretto2020lagrangian,kim2025DeepLDE}, which incorporates the objective function and penalties for equality and inequality constraint violations:
	\begin{equation*}
		\mathcal{L}(\bm{\phi}, \bm{\lambda}, \bm{\nu}, \bm{d})  = f(\bm{y},\bm{d}) + \bm{\lambda}^\top \left(\bm{g}(\bm{y},\bm{d})\right)^{+} + \bm{\nu}^\top |\bm{h}(\bm{y},\bm{d})|,
	\end{equation*}
	where $(\cdot)^{+}$ denotes $\max(\bm{0}, \cdot)$, $|\cdot|$ denotes the element-wise absolute value, $\bm{y} = \mathcal{M}_{\bm{\phi}}(\bm{d})$ is the output of the forward pass, and the dual variables $\bm{\lambda} \in \mathbb{R}^{\text{ineq}}$ and $\bm{\nu} \in \mathbb{R}^{\text{eq}}$ adjusted to penalize any constraint violations.
	The goal is to find the optimal NN weights $\bm{\phi}$ that minimize the loss taken over the dataset $\mathcal{S}$ of load profiles $\bm{d}$:
	\begin{equation*}
		\min_{\bm{\phi}} \max_{ \bm{\lambda} \geq 0, \bm{\nu} \geq 0} \frac{1}{|\mathcal{S}|} \sum_{\bm{d} \in  \mathcal{S} }  \mathcal{L}(\bm{\phi}, \bm{\lambda}, \bm{\nu}, \bm{d}).
	\end{equation*}
	The training is structured as a primal-dual algorithm via alternating gradient-based steps.
	\par
	For the primal update, the NN weights are adjusted using gradient descent on the Lagrangian: 
	\begin{equation}
		\label{eq:grad_for_phi}
		\bm{\phi} \leftarrow \bm{\phi} - \eta_{\bm{\phi}} \nabla_{\bm{\phi}} \mathcal{L}.
	\end{equation}
	For the dual update, the Lagrange multipliers are updated using gradient ascent, effectively increasing the penalties on violated constraints:
	\begin{align*}
		\bm{\lambda} &\leftarrow \bm{\lambda} + \eta_{\bm{\lambda}} \nabla_{\bm{\lambda}} \mathcal{L}(\bm{\phi}, \bm{\lambda}, \bm{\nu}), \\
		\bm{\nu} &\leftarrow \bm{\nu} + \eta_{\bm{\nu}} \nabla_{\bm{\nu}} \mathcal{L}(\bm{\phi}, \bm{\lambda}, \bm{\nu}),
	\end{align*}
	where $\eta_{\bm{\phi}}$, $\eta_{\bm{\lambda}}$, and $\eta_{\bm{\nu}}$ are the learning rates. This iterative process, detailed in Algorithm \ref{alg:primal_dual_training}, drives the NN to produce solutions that are not only low-cost but also feasible with respect to all physical laws and operational limits.
	
	\subsection{Solving Power Flow with a FDPF Fixed-Point Solver}
	\label{subsec:pf_solver}
	To strictly enforce the power flow equations $\bm{h}(\bm{y}, \bm{d}) = \bm{0}$, we classify the system buses into sets of slack buses $\mathcal{R}$, generator buses (without the slack bus) $\mathcal{G}$, and load buses $\mathcal{D}$. The decision vector $\bm{y}$ is partially predicted and partially derived. Specifically, we partition $\bm{y}$ into three subsets $\bm{x}$, $\bm{z}$, and $\tilde{\bm{z}}$ as illustrated in Figure~\ref{fig:fpl-framework}. The neural network $\mathcal{M}_{\bm{\phi}}(\bm{d})$ predicts only the independent controls $\bm{x}$ while the remaining dependent states are resolved via physical constraints.
	\par
	\textit{Prediction Variables.} $\bm{x} = [( P^{g}_{\mathcal{G}} )^\top, (V_{ \mathcal{G} \cup \mathcal{R}})^\top]^\top \in \mathbb{R}^{2|\mathcal{G}|+|\mathcal{R}|}$.
	The vector $\bm{x}$ is the direct output of the neural network $\mathcal{M}_{\bm{\phi}} (\bm{d})$, conditioned on the system load profile $\bm{d}$.
	\par 
	\textit{Completion Variables.} $\bm{z} = [(\theta_{\mathcal{G} \cup \mathcal{D}})^\top, (V_{ \mathcal{D}})^\top ]^\top \in \mathbb{R}^{|\mathcal{G}|+2|\mathcal{D}|}$.
	These are the dependent variables determined by the prediction variables $\bm{x}$ and the load profile $\bm{d}$ through the non-linear PF equations, namely the active power balance constraints \eqref{eq:pf_p} at buses $\mathcal{G}$ and the reactive power balance constraints \eqref{eq:pf_q} at buses $\mathcal{D}$. 
	\par 
	\textit{Post-completion Variables.} $\tilde{\bm{z}} = [ (P^{g}_{\mathcal{R}})^\top, (Q^{g}_{\mathcal{R} \cup \mathcal{G}})^\top ]^\top \in \mathbb{R}^{2|\mathcal{R}|+|\mathcal{G}|}$.
	Once $\bm{z}$ is found, the remaining system variables, \textit{i.e.}, the active power generation  $P^{g}_{\mathcal{R}}$ and the reactive power generation $Q^{g}_{\mathcal{R} \cup \mathcal{G}}$, can be explicitly computed from the power balance equations \eqref{eq:pf_p} at buses $\mathcal{R}$ and equations \eqref{eq:pf_q} at buses $\mathcal{R} \cup \mathcal{G}$ to complete the system state.
	\par 
	The process of finding the completion variables $\bm{z}$ given the prediction variables $\bm{x}$ is formulated as a fixed-point problem, $\bm{z} = T(\bm{z}, \bm{x})$, where $T$ denotes the FDPF iteration.
	We here omit the dependency on the load profile $\bm{d}$ for simplicity.
	FDPF~\cite{monticelliFastDecoupledLoad1990,stott1974FastDecoupledLoadFlow} is a highly efficient variant of the NR method that leverages key physical characteristics of high-voltage transmission systems, i.e., negligible lines’ series resistances and shunt reactances and small voltage angle differences~\cite{buasonAnalysisFastDecoupled2021}. 
	These physical properties allow the full PF Jacobian to be decoupled and approximated by two constant sparse matrices, $\bm{B}'$ and $\bm{B}''$.
	This decouples the PF problem into two smaller and independent linear systems that are solved iteratively:
	$$
	\notag
	\begin{aligned}
		\bm{\theta}^{(k+1)} = \bm{\theta}^{(k)} + [\bm{B}']^{-1} \left( {\Delta \bm{P}^{(k)}}/{\bm{V}^{(k)}} \right), \\
		\bm{V}^{(k+1)} = \bm{V}^{(k)} + [\bm{B}'']^{-1} \left( {\Delta \bm{Q}^{(k)}}/{\bm{V}^{(k)}} \right).
	\end{aligned}
	$$
	Since $\bm{B}'$ and $\bm{B}''$ are constant, they can be computed and factorized only once before the training and inference process begin, hence each FDPF iteration reduces to inexpensive triangular linear solves with the pre-factorization factors.

	\subsection{Fast Differentiable Physics-Aware Layer}
	\label{subsec:forward_pass}
	The main computational bottleneck of physics-aware OPF learning lies in the backward pass through the PF completion.
	From \Cref{subsec:pf_solver}, the loss $\mathcal{L}$ depends on $(\bm{x}, \bm{z}, \bm{\tilde z})$, the gradient with respect to the network parameters $\bm{\phi}$ follows from the chain rule:
	\begin{equation}
		\label{eq:loss_grad_for_phi}
		\frac{\partial \mathcal{L}}{\partial \bm{\phi}}=
		\left(
		\frac{\partial\mathcal{L}}{\partial \bm{x}}+
		\frac{\partial\mathcal{L}}{\partial \tilde{\bm{z}}}\frac{\partial \tilde{\bm{z}}}{\partial \bm{x}}+
		\left( 
		\frac{\partial\mathcal{L}}{\partial \bm{z}^{\star}} + \frac{\partial\mathcal{L}}{\partial \tilde{\bm{z}}}\frac{\partial \tilde{\bm{z}}}{\partial \bm{z}^{\star}} \right) \frac{\partial \bm{z}^{\star}}{\partial \bm{x}}
		\right)
		\frac{\partial \bm{x}}{\partial \bm{\phi}},
	\end{equation}
	which requires the critical Jacobian $\partial \bm{z}^{\star}/\partial \bm{x}$, representing the sensitivity of the PF solution $\bm{z}^{\star}$ to the NN's prediction $\bm{x}$.
		\begin{figure}
		\centering
		\includegraphics[width= \linewidth]{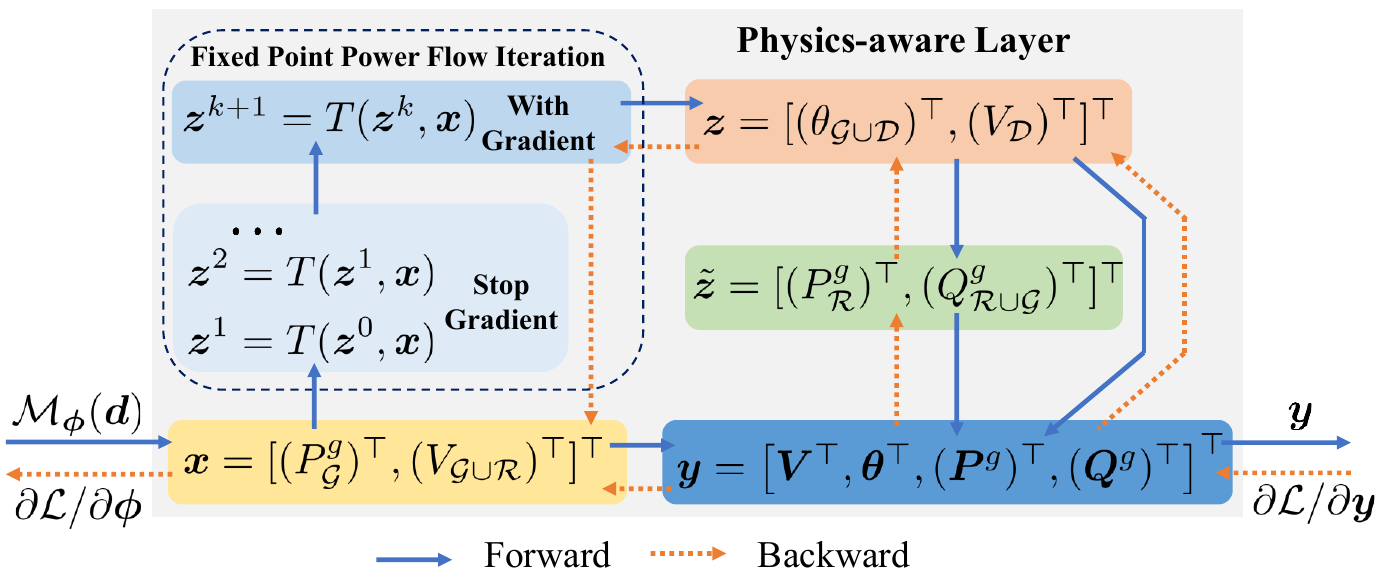}
		\caption{The overall architecture of the Fast Differentiable Physics-Aware Layer.
			This layer embeds a fixed-point PF solver that utilizes a rapid guide phase (Stop Gradient) to approximate the solution, followed by a differentiable refinement step (With Gradient) to efficiently compute gradients.}
		\label{fig:fpl-framework}
	\end{figure}
	Prior methods typically compute this implicit gradient by applying the Implicit Function Theorem (IFT)~\cite{dontchev2009implicit} to the static PF equations, $\bm{h}(\bm{z}, \bm{x}) = \bm{0}$ at the exact satisfaction point $\bm{z}^\star$.
	This yields $\frac{ \partial \bm{z}^\star}{\partial \bm{x}} = - \left[ \frac{ \partial h}{\partial \bm{z}^\star} \right]^{-1} \left(\frac{ \partial h} {\partial \bm{x}} \right)$, but necessitates forming and solving a large-scale linear system involving the PF Jacobian at each training step, creating a significant computational burden.
	Our framework introduces a more efficient and user-friendly alternative by following the contraction property of the FDPF solver originally established in~\cite{wuTheoreticalStudyConvergence1977}:
	\begin{lemma}[Local Contraction of FDPF Iteration~\cite{wuTheoreticalStudyConvergence1977}]
		\label{lemma:fdpf_contraction}
		$\qquad$
		The FDPF fixed-point operator, $\bm{z}^{k+1} = T(\bm{z}^k, \bm{x})$, is continuously differentiable around the power flow solution $\bm{z}^{\star}(\bm{x})$.
		For each fixed input $\bm{x}$, there exists a neighborhood $\mathcal{Z}$ of flat start point (one per unit voltage magnitude with the zero voltage angle) including the solution $\bm{z}^\star(\bm{x})$ such that:
		\newline
		(1) $T(\bm{z}, \bm{x}) \subseteq \mathcal{Z}, \forall \bm{z} \in \mathcal{Z}$;
		\newline
		(2) the Jacobian operator norm $\|J_{\bm{z}} T(\bm{z}, \bm{x})\|_{\mathrm{op}} \le \rho < 1, \forall \bm{z} \in \mathcal{Z}.\qquad$
	\end{lemma}
	This assumption implies that for a given input $\bm{x}$, there exists a neighborhood $\mathcal{Z}$ around the solution $\bm{z}^{\star}(\bm{x})$ such that for any iterate $\bm{z} \in \mathcal{Z}$, the next iterate $T(\bm{z}, \bm{x})$ remains in $\mathcal{Z}$. Furthermore, the Jacobian of the operator is bounded, $\|J_{\bm{z}} T(\bm{z}, \bm{x})\|_{\mathrm{op}} = \rho < 1, \forall \bm{z} \in \mathcal{Z}$, ensuring linear convergence to $\bm{z}^{\star}(\bm{x})$ for any initial point within this neighborhood.
	The contraction constant $\rho$ depends on system parameters such as the $R/X$ ratio of transmission lines and the load level, and typically remains sufficiently below $1$ in high-voltage transmission networks~\cite{wuTheoreticalStudyConvergence1977}.
	%\par
	Instead of relying on the static constraint equation, from \Cref{lemma:fdpf_contraction}, we derive the implicit gradient from the fixed point iterative dynamics 	$\bm{z}^{k+1}=T(\bm{z}^k, \bm{x})$, which converges to the solution $\bm{z}^{\star}= T(\bm{z}^\star, \bm{x})$ under standard assumptions~\cite{meiselApplicationFixedPointTechniques1970,wuTheoreticalStudyConvergence1977}. 
	The exact Jacobian ${\partial \bm{z}^{\star} }/{\partial \bm{x}}$ can again be found by applying the IFT at the $\bm{z}^{\star}-T(\bm{z}^{\star}, \bm{x}) = \bm{0}$:
	\begin{small}
		\begin{align}
			\label{eq:ift_fixed_point}
			J^{\ID}_{\bm{x}} \bm{z}^\star(\bm{x}) \triangleq\frac{\partial \bm{z}^\star}{\partial \bm{x}}
			&=
			\left[
			\bm{I} - \frac{\partial T (\bm{z}^{\star}, \bm{x})}{\partial \bm{z}}
			\right]^{-1}
			\left[
			\frac{\partial T (\bm{z}^{\star}, \bm{x})}{\partial \bm{x}}
			\right].
		\end{align}
	\end{small}However, forming and solving with the PF Jacobian at every training step is still computationally costly.
	In this work, our key innovation is to approximate the true implicit gradient by embedding the final fixed-point iteration into the AD graph and reducing implicit differentiation to a one-step operation that can be captured by the AD framework, as shown in Figure~\ref{fig:fpl-framework}.
	Concretely, we execute the $k$ iterations that produce $\bm{z}^k$ outside the computational graph and then insert only the final fixed-point iteration inside the AD computation graph:
	$$
	\bm{z}^{k+1}=T(\bm{z}^k, \bm{x}), \quad 
%	\bm{z}^\star = T(\bm{z}^\star, \bm{x}),
	$$
	where $\bm{z}^k$ is treated as a detached constant. The gradient of the output $\bm{z}^{k+1}$ with respect to the input $\bm{x}$ is then simply the partial derivative of $T$:
	$$
	J^{\OS}_{\bm{x}} \bm{z}^{k+1}\triangleq\frac{\partial T (\bm{z}^k, \bm{x})}{\partial \bm{x}}.
	$$
	This one-step Jacobian $J^{\text{OS}}$ serves as an approximation for the true Jacobian in \eqref{eq:ift_fixed_point}.
	The following lemma establishes a formal error bound for the one-step Jacobian approximation.
	%	The following lemma formally justifies the approximation quality.
	\begin{lemma}[One-step Jacobian Approximation Error~\cite{bolteOneStepDifferentiationIterative2023}]
		\label{lemma:os_id_error}
		Under \Cref{lemma:fdpf_contraction}, let $\bm{x} \mapsto T(\bm{z},\bm{x})$ be $L_T$-Lipschitz continuous and $\bm{z} \mapsto J_{\bm{x}} T(\bm{z}, \bm{x})$ be $L_J$-Lipschitz continuous in operator norm for all $\bm{z} \in \mathcal{Z}$. Then, for the power flow solution $\bm{z}^\star(\bm{x})$ and its $(k+1)$-th iterate $\bm{z}^{k+1}(\bm{x})$, the one-step Jacobian approximation error between $J^{\OS}_{\bm{x}} \bm{z}^{k+1}(\bm{x})$ and $J^{\mathrm{ID}}_{\bm{x}} \bm{z}^\star(\bm{x})$ satisfies:
		\[
		\left\| J^{\OS}_{\bm{x}} \bm{z}^{k+1}(\bm{x}) - J^{\mathrm{ID}}_{\bm{x}} \bm{z}^\star(\bm{x}) \right\|_{\mathrm{op}} \leq \frac{\rho L_T}{1-\rho} + L_J \|\bm{z}^{k} - \bm{z}^\star\|.
		\]
%		where $J^{\OS}_{\bm{x}} \bm{z}^{k+1}(\bm{x}) \triangleq \frac{\partial T(\bm{z}^{k}, \bm{x})}{\partial \bm{x}}$ and $J^{\mathrm{ID}}_{\bm{x}} \bm{z}^\star(\bm{x})$ is the true Jacobian of the fixed point solution with respect to $\bm{x}$.
	\end{lemma}
	\Cref{lemma:os_id_error} provides a formal bound on the approximation error, showing that the one-step Jacobian $J^{\OS}_{\bm{x}} \bm{z}^{k+1}$ converges to the true implicit Jacobian $J^{\ID}_{\bm{x}} \bm{z}^\star(\bm{x})$ as the forward iterate $\bm{z}^{k+1}$ approaches the fixed point $\bm{z}^\star$ and $\rho \to 0^{+}$.
	We use modern AD engines to compute $J^{\OS}_{\bm{x}} \bm{z}^{k+1}$ efficiently without unrolling the entire solver.
	We replace exact sensitivity Jacobian ${\partial \bm{z}^{\star}}/{\partial \bm{x}}$ in \eqref{eq:loss_grad_for_phi} with  $J^{\OS}_{\bm{x}} \bm{z}^{k+1}$, which leads to the following approximate parameter gradient:
	\begin{small}
		\begin{equation}
			\label{eq:grad_for_phi_by_os}
			\widehat{\frac{\partial \mathcal{L}}{\partial \bm{\phi}} }
			=
			\left(
			\frac{\partial\mathcal{L}}{\partial \bm{x}}+
			\frac{\partial\mathcal{L}}{\partial \tilde{\bm{z}}}\frac{\partial \tilde{\bm{z}}}{\partial \bm{x}}+
			\left( 
			\frac{\partial\mathcal{L}}{\partial \bm{z}^{k+1}} + \frac{\partial\mathcal{L}}{\partial \tilde{\bm{z}}}\frac{\partial \tilde{\bm{z}}}{\partial \bm{z}^{k+1}} 
			\right) 
			%			\frac{\partial T (\bm{z}^{k}, \bm{x})}{\partial \bm{x}}
			J^{\OS}_{\bm{x}} \bm{z}^{k+1}
			\right)
			\frac{\partial \bm{x}}{\partial \bm{\phi}}.
		\end{equation}
	\end{small}To theoretically justify the effectiveness of the proposed approximation strategy, in this work, we discuss the directional alignment between the true and approximate gradients with respect to the NN's parameters $\bm{\phi}$.
	%	In this work, we formally analyze the alignment between approximate and true gradients. 
	For notation brevity, we define the intermediate partial gradient with respect to $\bm{x}$ as:
	\begin{equation*}
		\partial_{\bm{x}}\mathcal{L}_{\mathrm{part}} \triangleq \frac{\partial\mathcal{L}}{\partial \bm{x}} + \frac{\partial\mathcal{L}}{\partial \tilde{\bm{z}}}\frac{\partial \tilde{\bm{z}}}{\partial \bm{x}},
	\end{equation*}
	and the intermediate total gradient with respect to $\bm{z}$ as:
	\begin{equation*}
		\partial_{\bm{z}}\mathcal{L}_{\mathrm{total}} \triangleq \frac{\partial\mathcal{L}}{\partial \bm{z}} + \frac{\partial\mathcal{L}}{\partial \tilde{\bm{z}}}\frac{\partial \tilde{\bm{z}}}{\partial \bm{z}}.
	\end{equation*}
	Furthermore, we rely on the following mild assumptions.
	
	\begin{assumption}
		\label{ass:lipschitz_loss_grad}
		There exist constants $L_{\bm{x}}, L_{\bm{z}}\ge 0$ such that for all $\bm{z},\bm{z}'\in\mathcal{Z}$,
		\begin{align*}
			\|\partial_{\bm{x}}\mathcal{L}_{\mathrm{part}}(\bm{x},\bm{z})
			-\partial_{\bm{x}}\mathcal{L}_{\mathrm{part}}(\bm{x},\bm{z}')\|
			&\le L_{\bm{x}}\|\bm{z}-\bm{z}'\|, \\
			\|\partial_{\bm{z}}\mathcal{L}_{\mathrm{total}}(\bm{x},\bm{z})
			-\partial_{\bm{z}}\mathcal{L}_{\mathrm{total}}(\bm{x},\bm{z}')\|
			&\le L_{\bm{z}}\|\bm{z}-\bm{z}'\|.
		\end{align*}
	\end{assumption}
	
	\begin{assumption}
		\label{ass:boundedness}
		There exist constants $C_{\bm{z}} > 0$ such that for all $\bm{z}\in\mathcal{Z}$,
		$\|\partial_{\bm{z}}\mathcal{L}_{\mathrm{total}}(\bm{x},\bm{z})\|\le C_{\bm{z}}$.
	\end{assumption}
	
	\begin{assumption}
		\label{ass:non_vanish_grad}
		There exists a constant $C_{\bm{g}} > 0$ such that the true gradient ${\partial \mathcal{L}}/{\partial \bm{\phi}}$ is non-vanishing, satisfying
		$
		\left\|
		{\partial \mathcal{L}}/{\partial \bm{\phi}}
		\right\|
		\ge C_{\bm{g}} .
		$
	\end{assumption}

\begin{assumption}
	\label{ass:op_norm_bound}
	There exists a constant $\sigma_{\bm{A}}, \sigma_{J}>0$ such that the neural network Jacobian $\bm{A} \triangleq \partial \bm{x}/\partial \bm{\phi}$ is bounded by $\|\bm{A}\|_{\mathrm{op}} \le \sigma_{\bm{A}}$, and the exact physical sensitivity Jacobian $J^{\mathrm{ID}}_{\bm{x}}\bm{z}^{\star}$ is bounded by $\|J^{\mathrm{ID}}_{\bm{x}}\bm{z}^{\star}\|_{\mathrm{op}}\le \sigma_{J}$.
\end{assumption}
			\begin{algorithm}[t]
		\caption{Forward Pass Fixed-Point Solver}
		\label{alg:hybrid_solver}
		%		\SetAlgoLined
		\KwIn{{Load profile $\bm{d}$, prediction variables $\bm{x}$, initial completion variables $\bm{z}^{0}$}}
		\KwOut{{Converged completion variables $\bm{z}^\star$}}
		Initialize $\bm{z} \gets \bm{z}^{0}$ \\
		%		\Begin{
			\AlgComment{stop\_gradient, fast-decoupled iterations}
			\For{$k=0$ to $K_{G}$}{
				$\bm{z}^{k+1} \gets \text{Guide\_Iteration}(\bm{z}^{k},\bm{x})$
			}
			\AlgComment{with\_gradient, final high-precision refinement}
			$\bm{z}^\star \gets \text{Refinement\_Step}(\bm{z}^{K_{G}},\bm{x})$ \\
			\Return $\bm{z}^\star$
			%		}
	\end{algorithm}
	Building on the Jacobian convergence in \Cref{lemma:os_id_error} and mild assumptions above, we further obtain a directional agreement guarantee stated below.
\begin{theorem}[Gradient Directional Alignment]
	\label{thm:gradient_direction_align}
	Assume Lemmas \ref{lemma:fdpf_contraction} and \ref{lemma:os_id_error} and Assumptions \ref{ass:lipschitz_loss_grad}--\ref{ass:op_norm_bound} hold.
	Let $d_0 \triangleq \|\bm{z}_0 - \bm{z}^{\star}\|$, $C_1 \triangleq \rho (L_{\bm{x}} + L_{\bm{z}} \sigma_{J}) + C_{\bm{z}} L_J$, and 
	$\varepsilon_{k}\triangleq\frac{\sigma_{\bm{A}}}{C_{\bm{g}}} \left(  \frac{\rho L_T C_{\bm{z}} }{1-\rho} 	+ \rho^k d_0 C_1 \right)$.
	Then, for every $k$ satisfying $\varepsilon_{k} < 1$, the cosine similarity between the true gradient $\frac{\partial \mathcal{L}}{\partial \bm{\phi}}$ and its approximation gradient $\widehat{\frac{\partial \mathcal{L}}{\partial \bm{\phi}} }$ satisfies:
	$$
	\left \langle 
	\frac{\partial \mathcal{L}}{\partial \bm{\phi}}, 
	\widehat{\frac{\partial \mathcal{L}}{\partial \bm{\phi}} } 
	\right \rangle
	\ge
	\frac{1-\varepsilon_{k}}{1+\varepsilon_{k}}
	> 0.
	$$
\end{theorem}
\begin{proof}
	The proof is provided in Appendix \ref{app:proof_gradient_direction_align}.
\end{proof}

\begin{remark}[Asymptotic Alignment Lower Bound]
	The error bound $\varepsilon_{k}$ can be decoupled into an irreducible asymptotic approximation error $\varepsilon_{\infty} \triangleq \frac{\sigma_{\bm{A}}}{C_{\bm{g}}} \frac{\rho L_T C_{\bm{z}} }{1-\rho}$ and a forward pass error $\Delta \varepsilon_{k} \triangleq \frac{\sigma_{\bm{A}}}{C_{\bm{g}}}\rho^k d_0 C_1$.
	As the number of fixed-point iterations $k \to \infty$, $\Delta \varepsilon_{k}$ decays geometrically to zero at the rate of $\rho^k$. Consequently, provided that $\rho$ is sufficiently small to ensure $\varepsilon_{\infty} < 1$, the asymptotic cosine similarity satisfies:
	$$
	\liminf_{k\to\infty}
	\cos\left( 
	\frac{\partial \mathcal{L}}{\partial \bm{\phi}}, \widehat{\frac{\partial \mathcal{L}}{\partial \bm{\phi}} } 
	\right)
	\ge
	\frac{1-\varepsilon_{\infty}}{1+\varepsilon_{\infty}} > 0.
	$$
	
	This result yields a crucial theoretical insight that the asymptotic alignment is characterized by the irreducible term $\varepsilon_{\infty}$.

\end{remark}
	This result formally establishes that, once the forward solver is sufficiently close to convergence, the one-step gradient is directionally aligned with the exact implicit gradient. 
	This ensures that it provides a valid descent direction for training.
	This approach integrates seamlessly with standard AD frameworks, circumventing the cumbersome and error-prone manual derivation of implicit differentiation procedures.
	The one-step differentiation schema of this method yields substantial computational savings by differentiating only a single, shallow computation.
		\begin{table}[t]
		\centering
		\caption{Complexity Analysis of Differentiation Methods. $n = |\mathcal{G}| + 2 |\mathcal{D}|$, $m = 2|\mathcal{G}|+|\mathcal{R}|$, and $\omega$ is the multiplicative overhead of gradient evaluation relative to function evaluation.}
		\label{tab:complexity_analysis}
		\resizebox{0.9\linewidth}{!}{
				\begin{tabular}{@{}llcc@{}}
					\toprule
					& \textbf{Method} & \textbf{Forward} & \textbf{Backward} \\ \midrule
					\multirow{3}{*}{\textbf{Time}} & Implicit Diff. & $\mathcal{O}(k n^3)$ & $\mathcal{O}(nm + n^2)$ \\
					& Single NR & $\mathcal{O}(K_G n^2 + n^3)$ & $\mathcal{O}(\omega n^2)$ \\
					& K-Step FDPF & $\mathcal{O}((K_{G}+K_{R}) n^2)$ & $\mathcal{O}(K_{R} \omega n^2)$ \\ \midrule
					\multirow{3}{*}{\textbf{Memory}} & Implicit Diff. & $\mathcal{O}(n^2)$ & $\mathcal{O}(n^2)$ \\
					& Single NR & $\mathcal{O}(n^2)$ & $\mathcal{O}(n^2)$ \\
					& K-Step FDPF & $\mathcal{O}(n^2)$ & $\mathcal{O}(n^2)$ \\ \bottomrule
				\end{tabular}
	}
	\end{table}
	However, the FDPF method exhibits linear convergence, and its contraction constant $\rho$ may not be small enough to satisfy $\varepsilon_{k} < 1$ in \Cref{thm:gradient_direction_align}. To address this, we can compose the operator $T$ with itself~\cite{bolteOneStepDifferentiationIterative2023}. If we define a new operator $\tilde{T} = T^{K}$, i.e., applying $T$ for $K$ iterations, this single composite operator will have a stronger contraction constant of $\tilde{\rho} = \rho^K$. By choosing a sufficient but not large $K$, we can make $\tilde{\rho}$ arbitrarily small, thus satisfying the theorem's condition. In this case, the one-step estimator applied to $\tilde{T}$ becomes a $K$-step estimator on $T$, where we differentiate through the final $K$ steps of the algorithm. This concept is central to our hybrid solver design, which we detail in the next subsection and quantify empirically in \Cref{subsec:impact_kr}.
	
	\subsection{Hybrid Iterative Solver for Accelerated Training}
	\label{subsec:hybrid_solver}
	The forward pass solver must be both fast for rapid training and accurate enough to satisfy the near-convergence assumptions of our differentiation approach. A key requirement is ensuring the operator used in the final differentiable step has a sufficiently small contraction constant $\rho$ to meet the condition $\varepsilon_{k} < 1$ in \Cref{thm:gradient_direction_align}. To achieve this, we design a two-phase hybrid solver that combines the strengths of different numerical methods.
	
	\textbf{Rapid Guide Phase with Non-Differentiable FDPF}.
	The first phase rapidly guides the completion variables $\bm{z}$ from an initial guess into the neighborhood of the correct solution. We employ a fixed number of $K_{G}$ FDPF iterations. Since the FDPF method uses constant, pre-factorized matrices, these iterations are computationally inexpensive. This phase is performed outside the AD computation graph, e.g., within a \texttt{torch.no\_grad()} context, incurring no memory overhead for backpropagation and efficiently bringing the solution close to the convergence point.
	
	\textbf{High-Precision Phase with Differentiable Refinement}.
	Following the guide phase, the second phase refines the solution to high precision using a differentiable operator with a strong contraction property. This is the only part of the solver executed within the gradient tracking context, e.g., \texttt{torch.enable\_grad()}. We consider two effective strategies for this refinement step:
	\subsubsection{Single NR Step} The NR method exhibits quadratic convergence near the solution, meaning its effective contraction constant $\rho$ is extremely close to zero. A single NR step thus serves as a powerful refinement that easily satisfies the condition in \Cref{thm:gradient_direction_align}. While highly effective, one NR step requires forming and factorizing the full PF Jacobian, which can be computationally intensive.
	\subsubsection{Multi-Step FDPF Block} As discussed previously, we can use a block of $K$ FDPF iterations as the differentiable operator. The composite operator $T^K$ has a contraction constant of $\rho^K$. By choosing a suitable $K$, we can ensure $\rho^K$ is small enough to satisfy the condition in \Cref{thm:gradient_direction_align}. This approach avoids the costly Jacobian formation of the NR step, relying only on repeated and efficient triangular linear solves with the pre-factored FDPF matrices.
	\par
	The entire hybrid procedure is summarized in \Cref{alg:hybrid_solver}. This architecture ensures the mathematical conditions for our efficient backpropagation method are met while leveraging domain-specific solvers for speed, leading to an accelerated and stable training process.
	
	\subsection{Time and Computation Cost Analysis}
	\label{subsec:time-analysis}
	We analyze the time and space complexities of the forward and backward processes for three strategies solving the PF equations: traditional Implicit Differentiation, and two variants of our proposed FPL-OPF framework distinguished by their refinement stage, i.e., employing a Single NR step versus a K FDPF steps.
	
	Let $n = |\mathcal{G}| + 2 |\mathcal{D}|$ be the dimension of the completion variables $\bm{z}$ and $m = 2|\mathcal{G}|+|\mathcal{R}|$ be the dimension of the prediction variables $\bm{x}$.
	Let $k$ be the number of forward solver iterations.
	Let $C_h$ denote the time cost of evaluating the PF residual function $\bm{h}(\bm{z},\bm{x})$, which is typically $\mathcal{O}(|\mathcal{N}|^2)$ for dense admittance matrices.
	Let $C_{J_{\bm{z}}}$ denote the time cost of constructing the Jacobian $J_{\bm{z}} \bm{h}(\bm{z},\bm{x})$ with respect to $\bm{z}$, which has a time complexity of $\mathcal{O}(n^2)$ due to its $n \times n$ size. Let $C_{J_{\bm{x}}}$ denote the time cost of constructing the Jacobian $J_{\bm{x}} \bm{h}(\bm{z},\bm{x})$ with respect to $\bm{x}$, which has a time complexity of $\mathcal{O}(n m)$ due to its $n \times m$ size.
	Let $C_T$ denote the time cost of a NR fixed-point iteration $\bm{z}^{k+1} = \bm{z}^{k} - [J_{\bm{z}} \bm{h}(\bm{z}^{k}, \bm{x})]^{-1} \bm{h}(\bm{z}^{k}, \bm{x})$. Typically, the NR iteration costs $\mathcal{O}(n^2)$ when using pre-factored matrices, while the factorization for $J_{\bm{z}} \bm{h}(\bm{z}^{k}, \bm{x})$ costs $\mathcal{O}(n^3)$. Let $\omega > 0$ be the multiplicative overhead of gradient evaluation relative to function evaluation, which typically satisfies $\omega \leq 5$ by the cheap gradient principle~\cite{baur1983complexitypartialderivatives}. For simplicity, we do not account for batching effects.
	
	\textbf{Single NR}. The forward pass combines an FDPF iteration with a single NR refinement step.
	The FDPF phase performs $K_G$ iterations, where each iteration computes the active and reactive power mismatches, involving two residual evaluations with a total costing $C_h$, followed by solving two reduced linear systems using pre-factored decomposition of $\bm{B}'$ and $\bm{B}''$ with reduced dimensions approximately $n/2$.
	The pre-factorization costs $\mathcal{O}((n/2)^3)$ but is performed once outside the iteration procedure, while each solve costs $\mathcal{O}((n/2)^2)$.
	Thus, the FDPF forward time complexity is $\mathcal{O}(K_G ( C_h + (n/2)^2))$, which simplifies to $\mathcal{O}(K_G n^2)$ under the assumption that $C_h \leq n^2$.
	After the FDPF loop, a single NR step follows at cost $\mathcal{O}(n^3)$ since the need for a factorization. The total forward time complexity is therefore $\mathcal{O}(K_G n^2 + n^3)$. Space complexity is $\mathcal{O}(n^2)$ for storing the Jacobian and its factors, plus $\mathcal{O}(n)$ for the iterates.
	
	The backward propagation computes the gradient of the loss $\mathcal{L}$ with respect to inputs $\bm{x}$ by differentiating through the single NR update,  the FDPF output $\bm{z}^{k}$ detached from the computation graph. This leads to a backward time complexity of $\mathcal{O}(\omega C_T)$ since the cheap gradient principle.
	Space complexity remains $\mathcal{O}(n^2)$ for the stored Jacobian.
	
	\textbf{Implicit Differentiation} \cite{donti2021dc3,kim2025DeepLDE}. The forward pass performs full NR iterations on the reduced system.
	Total forward time complexity is $\mathcal{O}(k (n^3 + n^2))$ up to $k$ iterations, which simplifies to $\mathcal{O}(k n^3)$ when inversion dominates.
	Space complexity is $\mathcal{O}(n^2)$ for storing the final Jacobian and its inverse.
	The backward pass uses implicit differentiation on the PF equation $\bm{h}(\bm{z}^\star,\bm{x}) = \bm{0}$, yielding:
	$$\nabla_{\bm{x}} \mathcal{L} = -\left[ J_{\bm{x}} \bm{h}(\bm{z}^\star, \bm{x}) \right]^\top \bm{u} \;\; \text{where} \;\;
	\left[ J_{\bm{z}} \bm{h}(\bm{z}^\star,\bm{x}) \right]^\top \bm{u} = \nabla_{\bm{z}} \mathcal{L}.$$
	This requires computing $J_{\bm{x}} \bm{h}(\bm{z}^\star,\bm{x})$ at cost $\mathcal{O}(nm)$, and solving the linear system at cost $\mathcal{O}(n^2)$ by reusing the factorization from the final forward iteration. Total backward time complexity is $\mathcal{O}(nm + n^2)$, with space complexity $\mathcal{O}(n^2)$ for stored Jacobians and inverses.
	
	\textbf{K-Step FDPF}. The forward pass consists of $K_{G}$ guide FDPF steps and $K_{R}$ differentiable FDPF steps. Since FDPF matrices are constant and pre-factored, each iteration costs only $\mathcal{O}(n^2)$. The total forward time is therefore $\mathcal{O}((K_{G}+K_{R})n^2)$. The backward pass differentiates through the final $K_{R}$ FDPF iterations. This is equivalent to unrolling the loop $K_{R}$ times, with each step backpropagation costing $\mathcal{O}(\omega n^2)$. The total backward time is $\mathcal{O}(K_{R}\omega n^2)$. The main overhead is storing these matrix factors, leading to a space complexity of $\mathcal{O}(n^2)$, which does not grow significantly during backpropagation through the $K_{R}$ steps.
	
	The complexity analysis is summarized in \Cref{tab:complexity_analysis}. The comparison reveals the significant advantages of the K-Step FDPF refine strategy. Both the ID and Single NR refine strategy suffer from an $\mathcal{O}(n^3)$ time complexity in the forward pass due to Jacobian factorization for the NR steps. In contrast, the K-Step FDPF method avoids this bottleneck entirely, achieving a much more favorable $\mathcal{O}(n^2)$ forward time complexity. While its backward pass cost scales with the number of unrolled steps $K_{R}$, this is a tunable hyperparameter that can be kept small in practice. For large-scale power systems where $n$ is large, the K-Step FDPF method can offer the most computationally efficient choice for both forward and backward passes without compromising the theoretical guarantees for gradient validity.

	\section{Experiments}
	\label{sec:experiments}
	
	%		eq max mismatch, eq mean mismatch, Eq. Viol. Num.
	\begin{table*}[htbp]
		\centering
		\caption{Performance comparison of different learning methods on the IEEE 57-bus, PEGASE 89-bus, IEEE 118-bus, and NESTA 189-bus systems. Bold numbers indicate the best result for each metric, excluding the MIPS Solver baseline.}
		\label{tab:performance_comparison}
		\resizebox{0.85\linewidth}{!}{ 
		\begin{tabular}{lcccccccc}
			\toprule
			Method & \makecell[c]{Eq. Mean\\Mismatch} & \makecell[c]{Eq. Max.\\Mismatch} & \makecell[c]{Eq. Viol.\\Num.} & \makecell[c]{Ineq. Mean\\Mismatch} & \makecell[c]{Ineq. Max\\Mismatch} & \makecell[c]{Ineq. Viol.\\Num.} & \makecell[c]{Objective\\Cost (\$)} & \makecell[c]{Objective\\Gap (\%)} \\
			\midrule
			\multicolumn{9}{c}{\textbf{IEEE 57-bus system}} \\
			\midrule
			MIPS Solver & 0.00e+00 & 0.00e+00 & 0.00e+00 & 0.00e+00 & 0.00e+00 & 0.00e+00 & 41659.62 & 0.00\% \\
			\cline{2-9}
			DC3 & 3.37e-08 & 1.97e-06 & \textbf{0.00e+00} & 3.64e-07 & 1.10e-04 & \textbf{0.00e+00} & 41696.90 & 0.09\% \\
			OPF-DNN & 0.0670 & 0.7387 & 85.97 & \textbf{0.00e+00} & \textbf{0.00e+00} & \textbf{0.00e+00} & \textbf{39725.79} & \textbf{-4.39\%} \\
			DeepLDE & \textbf{1.70e-15} & \textbf{1.47e-14} & \textbf{0.00e+00} & 3.69e-07 & 1.12e-04 & \textbf{0.00e+00} & 41672.69 & 0.03\% \\
			FPL-OPF (NR) & 4.16e-11 & 6.39e-10 & \textbf{0.00e+00} & 4.49e-07 & 1.35e-04 & \textbf{0.00e+00} & 41670.51 & 0.02\% \\
			FPL-OPF (K-Decoupling) & 2.85e-09 & 4.55e-08 & \textbf{0.00e+00} & 4.52e-07 & 1.37e-04 & \textbf{0.00e+00} & 41674.59 & 0.03\% \\
			\midrule
			\multicolumn{9}{c}{\textbf{PEGASE 89-bus system}} \\
			\midrule
			MIPS Solver & 0.00e+00 & 0.00e+00 & 0.00e+00 & 0.00e+00 & 0.00e+00 & 0.00e+00 & 5794.07 & 0.00\% \\
			\cline{2-9}
			DC3 & 6.18e-06 & 2.07e-04 & \textbf{0.00e+00} & 2.06e-04 & 0.0507 & 4.17 & 5799.94 & 0.10\% \\
			OPF-DNN & 2.42 & 32.56 & 177.00 & 0.4046 & 130.60 & 3.00 & \textbf{3223.58} & \textbf{-44.29\%} \\
			DeepLDE & 9.54e-09 & 1.14e-06 & \textbf{0.00e+00} & 1.12e-04 & 0.0333 & 2.24 & 5807.48 & 0.23\% \\
			FPL-OPF (NR) & 5.63e-09 & 4.10e-07 & \textbf{0.00e+00} & 1.13e-04 & 0.0333 & 2.25 & 5807.11 & 0.23\% \\
			FPL-OPF (K-Decoupling) & \textbf{1.37e-09} & \textbf{6.22e-08} & \textbf{0.00e+00} & \textbf{1.11e-04} & \textbf{0.0327} & \textbf{2.23} & 5807.35 & 0.23\% \\
			\midrule
			\multicolumn{9}{c}{\textbf{IEEE 118-bus system}} \\
			\midrule
			MIPS Solver & 0.00e+00 & 0.00e+00 & 0.00e+00 & 0.00e+00 & 0.00e+00 & 0.00e+00 & 129673.19 & 0.00\% \\
			\cline{2-9}
			DC3 & 6.41e-06 & 6.49e-04 & \textbf{0.00e+00} & 5.02e-07 & 4.10e-04 & 0.0180 & 130090.96 & 0.32\% \\
			OPF-DNN & 0.0245 & 0.3829 & 141.22 & \textbf{0.00e+00} & \textbf{0.00e+00} & \textbf{0.00e+00} & \textbf{126730.01} & \textbf{-2.24\%} \\
			DeepLDE & \textbf{7.06e-14} & \textbf{6.05e-12} & \textbf{0.00e+00} & 2.53e-08 & 2.09e-05 & 0.0010 & 129749.35 & 0.06\% \\
			FPL-OPF (NR) & 2.77e-12 & 1.70e-10 & \textbf{0.00e+00} & 3.69e-08 & 3.04e-05 & 0.0010 & 129748.15 & 0.06\% \\
			FPL-OPF (K-Decoupling) & 2.59e-09 & 1.38e-07 & \textbf{0.00e+00} & 9.92e-08 & 8.18e-05 & 0.0020 & 129757.46 & 0.06\% \\
			\midrule
			\multicolumn{9}{c}{{\textbf{NESTA 189-bus system}}} \\
			\midrule
			MIPS Solver & 0.00e+00 & 0.00e+00 & 0.00e+00 & 0.00e+00 & 0.00e+00 & 0.00e+00 & 849.67 & 0.00\% \\
			\cline{2-9}
%			DC3 & 0.6828 & 26.54 & 311.64 & 0.6331 & 111.29 & 38.94 & 901.67 & 6.34\% \\
			OPF-DNN & 0.6828 & 26.54 & 311.64 & 0.6331 & 111.29 & 38.94 & 901.67 & 6.34\% \\
			DeepLDE & \textbf{4.09e-14} & \textbf{3.09e-12} & \textbf{0.00e+00} & 5.44e-05 & 0.0252 & 1.17 & 988.50 & 16.43\% \\
			FPL-OPF (NR) & 6.53e-10 & 4.34e-08 & \textbf{0.00e+00} & 1.54e-06 & 6.58e-04 & 0.0215 & 885.60 & 4.30\% \\
			FPL-OPF (K-Decoupling) & 1.35e-11 & 1.41e-09 & \textbf{0.00e+00} & \textbf{5.21e-07} & \textbf{2.81e-04} & \textbf{0.0075} & \textbf{861.75} & \textbf{1.48\%}
			\\ 	
			\bottomrule
		\end{tabular}}
	\end{table*}

	\begin{figure*}[t]    
		\centering
		
		\begin{subfigure}{0.245 \textwidth}
			% 请先 usepackage{subcaption}
			% subfigure 环境用于在 figure* 环境中插入子图。 
			% [b] 指定子图的基线对齐方式为底部对齐
			% {0.48\textwidth} 设置子图的宽度为页面宽度的 48%
			
			\centering
			\includegraphics[width=  \textwidth]{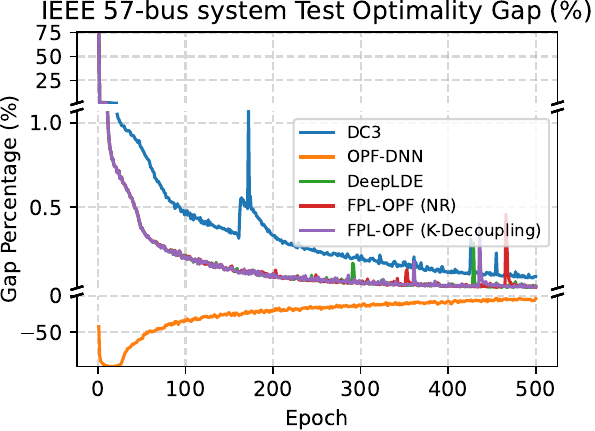}
			\caption{\footnotesize Optimality comparison at IEEE 57-bus system.}
			\label{fig:opt_gap_57}
		\end{subfigure}
		\begin{subfigure}[b]{0.245 \textwidth}
			\centering
			\includegraphics[width=  \textwidth]{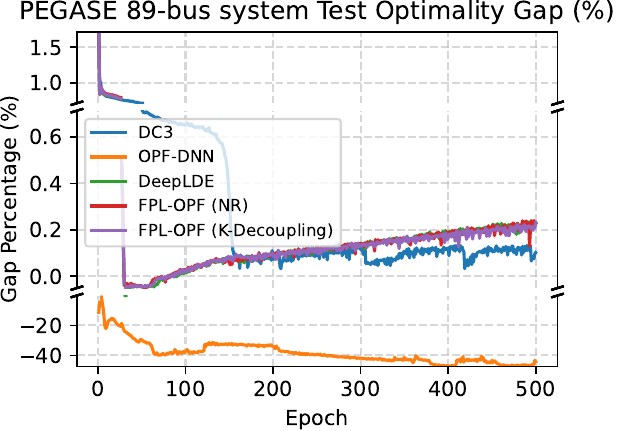}
			\caption{\footnotesize Optimality comparison at PEGASE 89-bus system.}
			
			\label{fig:opt_gap_89}
		\end{subfigure}
		\begin{subfigure}[b]{0.245 \textwidth}
			\centering
			\includegraphics[width=  \textwidth]{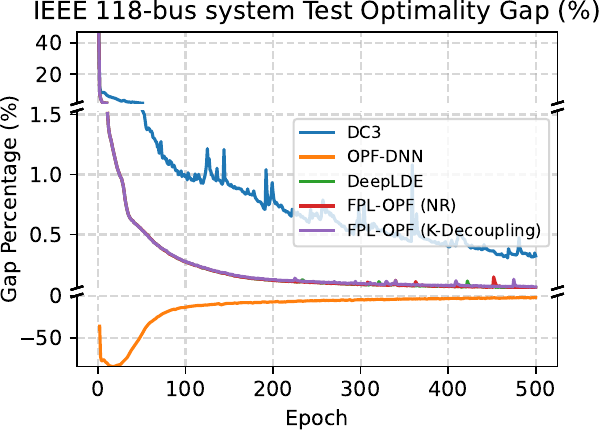}
			\caption{\footnotesize Optimality comparison at IEEE 118-bus system.}
			\label{fig:opt_gap_118}
		\end{subfigure}
				\begin{subfigure}[b]{0.245 \textwidth}
			\centering
			\includegraphics[width=  \textwidth]{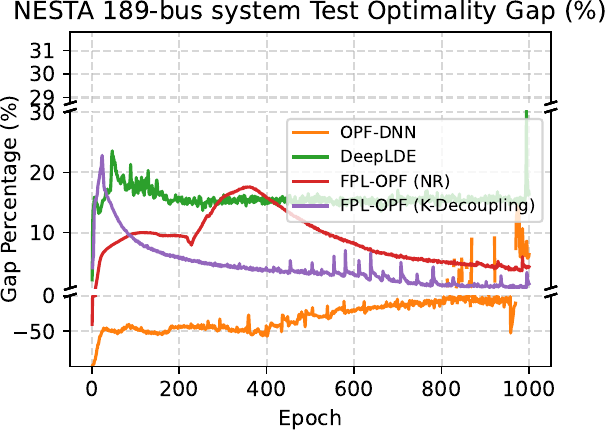}
			\caption{\footnotesize Optimality comparison at the NESTA 189-bus system.}
			\label{fig:opt_gap_189}
		\end{subfigure}
		\caption{Optimality comparison on (a) IEEE 57-bus, (b) PEGASE 89-bus, (c) IEEE 118-bus, and (d) NESTA 189-bus systems. 
%			Solid curves denote the mean optimality gap, and shaded bands denote one standard deviation over the test set.
			}    
		\label{fig:opt_gap}
	\end{figure*}

	\begin{figure*}[t]    
		\centering
		
		\begin{subfigure}[b]{0.22 \textwidth}
			% 请先 usepackage{subcaption}
			% subfigure 环境用于在 figure* 环境中插入子图。 
			% [b] 指定子图的基线对齐方式为底部对齐
			% {0.48\textwidth} 设置子图的宽度为页面宽度的 48%
			
			\centering
			\includegraphics[width= \textwidth]{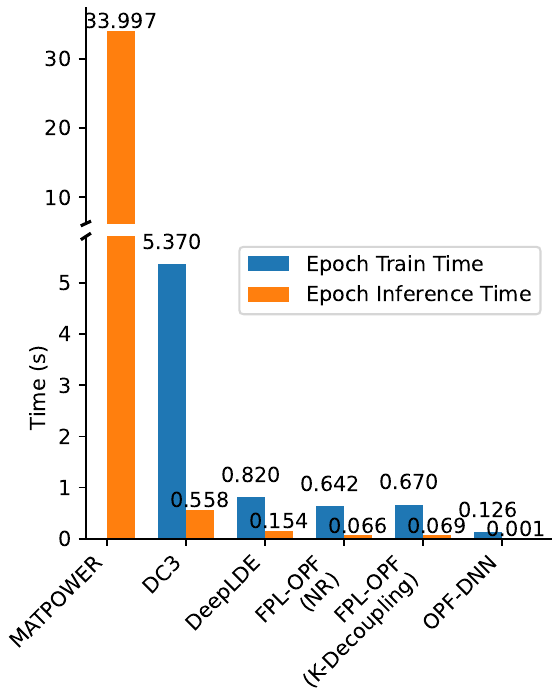}
			\caption{\footnotesize Time comparison at IEEE 57-bus system.}
			\label{fig:time_case57}
			
		\end{subfigure}
		\begin{subfigure}[b]{0.22 \textwidth}
			\centering
			\includegraphics[width= \textwidth]{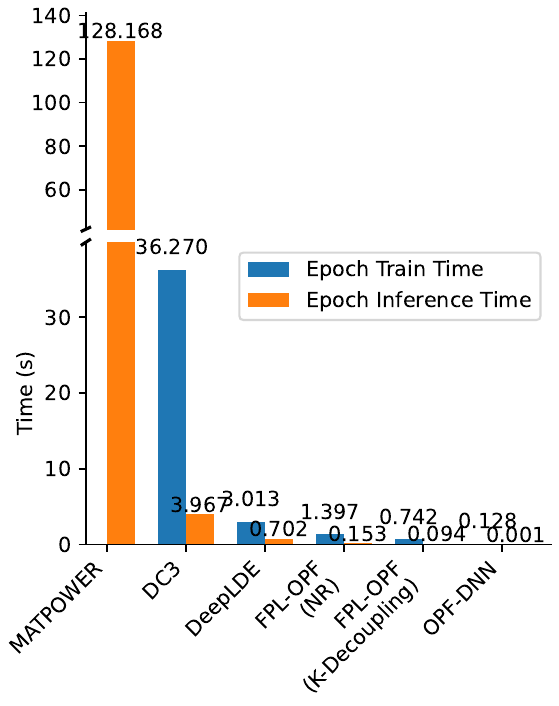}
			\caption{\footnotesize Time comparison at PEGASE 89-bus system.}
			
			\label{fig:time_case89}
		\end{subfigure}
		\begin{subfigure}[b]{0.22 \textwidth}
			\centering
			\includegraphics[width= \textwidth]{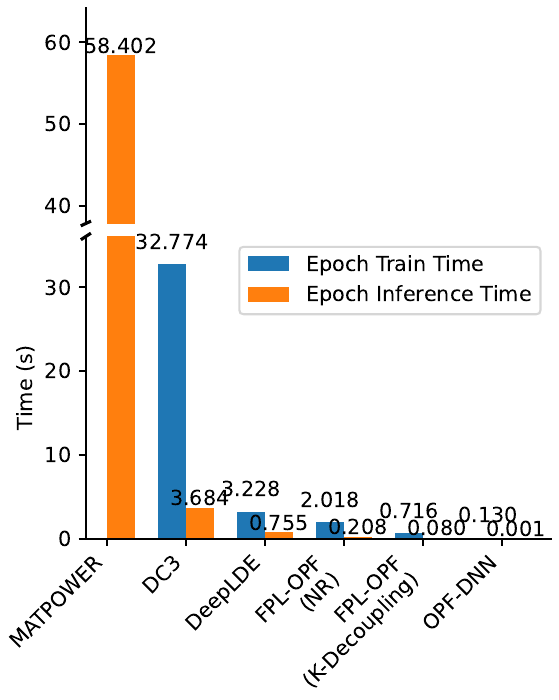}
			\caption{\footnotesize Time comparison at IEEE 118-bus system.}
			\label{fig:time_case118}
		\end{subfigure}
				\begin{subfigure}[b]{0.22 \textwidth}
			\centering
			\includegraphics[width= \textwidth]{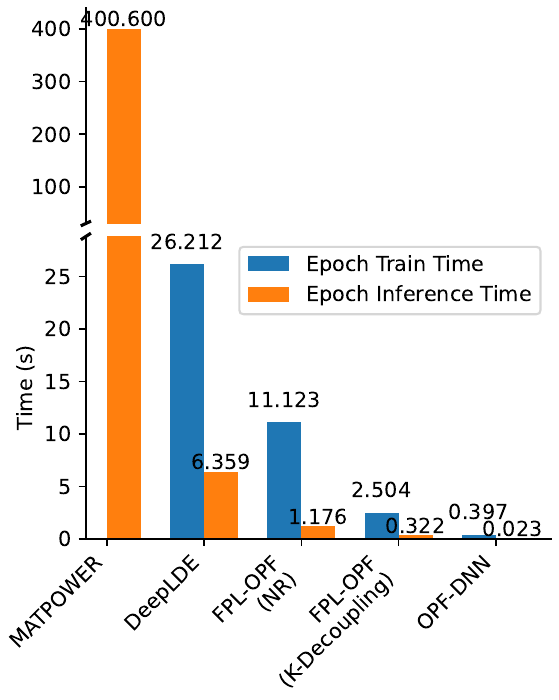}
			\caption{\footnotesize Time comparison at the NESTA 189-bus system.}
			\label{fig:time_case189}
		\end{subfigure}
		\caption{Time-efficiency comparison on (a) IEEE 57-bus, (b) PEGASE 89-bus, (c) IEEE 118-bus, and (d) NESTA 189-bus systems. Blue bars report training time per epoch, and orange bars report inference time on the test set.}    
		\label{fig:time_comparision}
	\end{figure*}
	To evaluate the proposed FPL-OPF framework, we consider four benchmark systems: the IEEE 57-bus, PEGASE 89-bus, and IEEE 118-bus from MATPOWER~\cite{zimmerman2010matpower} and NESTA 189-bus from NESTA~\cite{coffrin2014nesta}. These benchmarks cover increasingly challenging AC-OPF instances and enable us to assess both solution quality and computational scalability.
	
	\subsection{Experimental Setup}
	\label{subsec:exp_setup}
	
	\begin{table}[htbp]
		\centering
		\caption{{Model architectures and training hyperparameter settings}}
		\label{tab:hyperparameter_settings}
		\resizebox{0.95\linewidth}{!}{
		\begin{tabular}{@{}lcccc@{}}
			\toprule
			Power Grid                          & IEEE 57-bus & PEGASE 89-bus & IEEE 118-bus & NESTA 189-bus \\ \midrule
			$|\mathcal{G}|$                     & 7           & 12            & 54           & 35            \\
			$|\mathcal{E}|$                     & 80          & 210           & 186          & 206           \\
			Eq. Num.                            & 114         & 178           & 236          & 378           \\
			Ineq. Num.                          & 302         & 646           & 824          & 930           \\
			Hidden Layers                       & 200-200     & 300-300       & 200-200      & 4096-4096     \\
			Activation Function                 & ELU         & ELU           & ELU          & ELU           \\
			Learning Rate $\eta_{\bm{\phi}}$    & 1e-3        & 1e-3          & 1e-3         & 1e-4          \\
			Learning Rate $\eta_{\bm{\lambda}}$ & 0.1         & 0.01          & 0.01         & 0.01          \\
			Learning Rate $\eta_{\bm{\nu}}$     & 0.5         & 0.05          & 0.05         & 0.05          \\
			Optimizer                           & Adam        & Adam          & Adam         & Adam          \\
			Outer Iterations $\mathcal{T}$      & 20          & 20            & 20           & 40            \\
			Inner Iterations $\mathcal{I}$      & 25          & 25            & 25           & 25            \\
			Batch Size                          & 200         & 200           & 200          & 200           \\
			Epochs                              & 500         & 500           & 500          & 1000           \\ \bottomrule
		\end{tabular}
		}
	\end{table}
	
	\begin{table}[]
		\centering
		\caption{{Detailed configurations of the power flow solvers used in the forward pass. ‘‘N/A’’ indicates that this configuration item does not apply to the method.}}
		\label{tab:detailed_pf_solver_setting}
		\resizebox{0.95\linewidth}{!}{
\begin{tabular}{@{}llccccc@{}}
	\toprule
	\multicolumn{1}{l}{Power Grid}                                                                       & \multicolumn{1}{l}{Method} & Solver & \makecell{Total \\ Iterations} & \makecell{Guide \\ $K_G$} & \makecell{Refinement \\ $K_R$} & \makecell{Stopping \\ Tolerence} \\ 
	\midrule
	\multirow{3}{*}{\begin{tabular}[c]{@{}l@{}}IEEE 57-bus\\ PEGASE 89-bus \\ IEEE 118-bus\end{tabular}} & DC3, DeepLDE               & NR     & 10                           & N/A         & N/A              & $<10^{-5}$                     \\
	& FPL-OPF (NR)               & Hybrid & 10                           & 9 FDPF      & 1 NR             & $<10^{-5}$                     \\
	& FPL-OPF (K-Decoupling)     & Hybrid & 12                           & 8 FDPF      & 4 FDPF           & $<10^{-5}$                     \\ \midrule
	\multirow{3}{*}{NESTA 189-bus}                                                                       & DeepLDE                    & NR     & 18                           & N/A         & N/A              & $<10^{-5}$                     \\
	& FPL-OPF (NR)               & Hybrid & 18                           & 17 FDPF     & 1 NR             & $<10^{-5}$                     \\
	& FPL-OPF (K-Decoupling)     & Hybrid & 18                           & 10 FDPF     & 8 FDPF           & $<10^{-5}$                     \\ \bottomrule
\end{tabular}
	}
	\end{table}

	\subsubsection{Datasets and Load Profiles}
	We generate feasible load-profile datasets by perturbing the base operating points of each benchmark.
%	For the IEEE 57-bus, PEGASE 89-bus, and IEEE 118-bus systems, we sample 5{,}000 instances per case, using an 80\%/20\% train-test split and independent active and reactive demand perturbations drawn uniformly from $[80\%,120\%]$ of the nominal load.
%	For the NESTA 189-bus system, we generate 10{,}000 feasible samples with both active and reactive demands perturbed within $[90\%,110\%]$ of the nominal load.
For the IEEE 57-bus, PEGASE 89-bus, IEEE 118-bus, and NESTA 189-bus systems, we use an 80\% training and 20\% testing split, sampling 5{,}000 feasible samples for the first three systems and 10{,}000 feasible samples for the NESTA 189-bus system, with independent active and reactive demand perturbations drawn uniformly from $[80\%,120\%]$ of the nominal load for the first three systems and from $[90\%,110\%]$ for the NESTA 189-bus system.
	All experiments are implemented in PyTorch and run on an NVIDIA RTX 4090D GPU with an Intel i9-12600K CPU.
	\subsubsection{Baselines}
	To evaluate the effectiveness of the proposed approach, we compare different variants of FPL-OPF against several established benchmarks:
	\par
		 \textbf{MIPS Solver}~\cite{mips2007}: A primal dual interior point method in MATPOWER for solving non-linear programming problems. It serves as the ground-truth benchmark for feasibility and optimality.
\par
		 \textbf{OPF-DNN}~\cite{fiorettoPredictingACOptimal2020,fioretto2020lagrangian}: Unsupervised version of \cite{fiorettoPredictingACOptimal2020,fioretto2020lagrangian}  that directly optimizes the Lagrangian function using a penalty-based approach, without incorporating any explicit physics-informed layers.
\par
		 \textbf{DC3}~\cite{donti2021dc3}: A prominent physics-informed unsupervised method that embeds power flow physics via implicit differentiation.
		%		\item \textbf{FPL-OPF NR}: This variant employs the classical NR method for the forward fixed-point iterations. It serves as a high-precision baseline with quadratic convergence and a higher computational cost due to the factorization of the full Jacobian matrix.
\par
		\textbf{DeepLDE}~\cite{kim2025DeepLDE}: This method follows the DC3's exact implicit differentiation but without the DC3's correction mechanism. DeepLDE handles inequality constraints through a primal-dual learning scheme.
\par
		 \textbf{FPL-OPF (NR)}: This variant utilizes the FDPF solver for the rapid guide phase to bring the solution near convergence, followed by a single differentiable NR step for the refinement phase. 
\par
		\textbf{FPL-OPF (K-Decoupling)}: This variant utilizes FDPF for the guide phase and employs $K_{R}$ FDPF iterations for the differentiable refinement phase. By maintaining strict $\mathcal{O}(n^2)$ complexity for both forward and backward passes.
\par
	We train all methods for 500 epochs on the IEEE 57-bus, PEGASE 89-bus, and IEEE 118-bus systems, and for 1000 epochs on the NESTA 189-bus system. Model architectures and training hyperparameters are detailed in Table~\ref{tab:hyperparameter_settings}.
	Note that the DC3 baseline is excluded from the evaluations on the large NESTA 189-bus system, since the inequality correction mechanism in DC3 requires excessive memory allocations that exceed the capacity of an NVIDIA RTX 4090D GPU under a batch size of 200. This memory bottleneck becomes even more severe on larger grids. As detailed in Appendix~\ref{sec:appendix_case500}, all Newton-Raphson and exact implicit differentiation baselines encounter Out-Of-Memory failures on the PGLib 500-bus system, leaving FPL-OPF (K-Decoupling) as the only available learning-based approach.
	\subsubsection{Power Flow Solver Settings}
	To ensure a fair comparison, we align the stopping criterion and the total power flow iteration budget within each test grid. 
	The detailed settings are summarized in Table~\ref{tab:detailed_pf_solver_setting}.
	Specifically, the maximum iteration budget for exact implicit differentiation methods is rigorously justified by an ablation study detailed in Appendix~\ref{sec:appendix_ablation_deeplde_cap}.
	
	\subsubsection{Evaluation Metrics}
	We evaluate model performance on the training and test sets using a suite of metrics covering feasibility, optimality, and computational efficiency:
	\par
	\textbf{Constraint Satisfaction}: Evaluated via the mean and maximum mismatch for both equality and inequality constraints, along with the average number of violations exceeding a tolerance of $10^{-4}$.
 	\par
	\textbf{Optimality}: Measured by the average objective cost and the relative optimality gap on the test set, calculated as $100\% \times \frac{1}{|\mathcal{S}|} (\hat{\text{O}} - \text{O}^{\star}) / \text{O}^{\star}$, where $\hat{\text{O}}$ and $\text{O}^{\star}$ denote the optimal cost given by the NN and MIPS, respectively.
	\par
	\textbf{Computational Efficiency}: Measured by the total training time on the training set and inference time on the test set per epoch.
		
		\subsection{Results and Analysis}
		\label{subsec:results_analysis}
		
		\subsubsection{Training Convergence}
		Figure~\ref{fig:opt_gap} reports the test optimality gap curves on all four systems.
		The more fine-grained analysis of the training dynamics for these benchmarks are provided in Appendix~\ref{sec:appendix_exp_results}.
		On the IEEE 57-bus, PEGASE 89-bus, and IEEE 118-bus systems, the physics-aware methods, including DC3, DeepLDE, and both FPL-OPF variants, converge stably and typically approach their final gaps within roughly the first 200 epochs.
		The NESTA 189-bus case is substantially more challenging.
		Even so, FPL-OPF (K-Decoupling) continues to decrease the gap throughout training and reaches 1.48\% after 1000 epochs, outperforming FPL-OPF (NR) at 4.30\% and DeepLDE at 16.43\%.
		The numerical instability may be the reason behind the degraded performance of exact implicit differentiation methods like DeepLDE on this complex system.
		Specifically, the required Jacobian-inverse computation can become highly unstable when encountering ill-conditioned Jacobian matrices, which severely impairs the training dynamics \cite{bai2021stabilizing, gengTrainingImplicitModels2022}.
		In contrast, OPF-DNN remains highly unstable and exhibits large negative gaps caused by severe infeasibility, confirming the importance of embedding physical principles into the learning pipeline.
					\begin{figure}[t]    
			\centering
			\begin{subfigure}[t]{0.22 \textwidth}
				% 请先 usepackage{subcaption}
				% subfigure 环境用于在 figure* 环境中插入子图。 
				% [b] 指定子图的基线对齐方式为底部对齐
				% {0.48\textwidth} 设置子图的宽度为页面宽度的 48%
				
				\centering
				\includegraphics[width=  \textwidth]{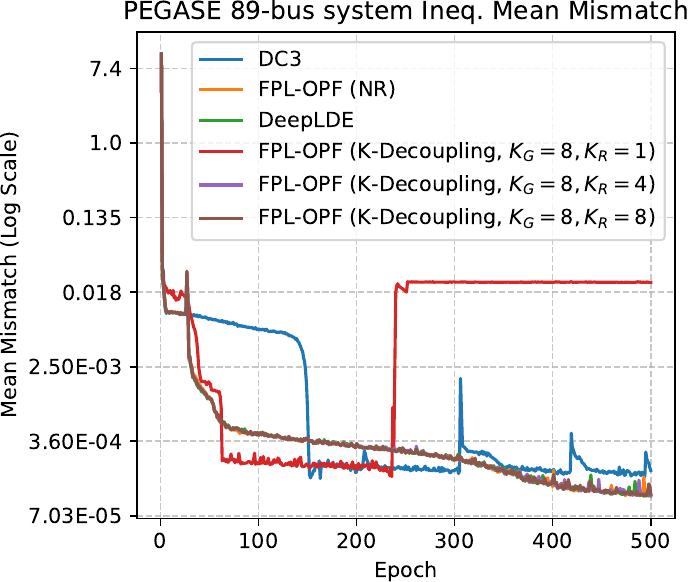}
				\caption{\footnotesize Inequality Mean Mismatch during training. }
				\label{fig:fpl_ablation_ineq_mean_mismatch}
			\end{subfigure}
			\begin{subfigure}[t]{0.22 \textwidth}
				\centering
				\includegraphics[width=  \textwidth]{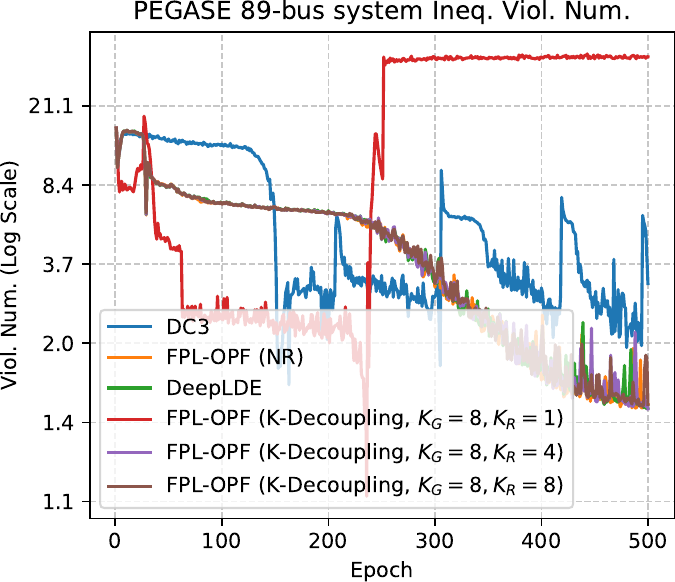}
				\caption{\footnotesize Inequality Violation Number during training.  }
				\label{fig:fpl_ablation_ineq_viol_num}
			\end{subfigure}
			\begin{subfigure}[t]{0.22 \textwidth}
				\centering
				\includegraphics[width= \textwidth]{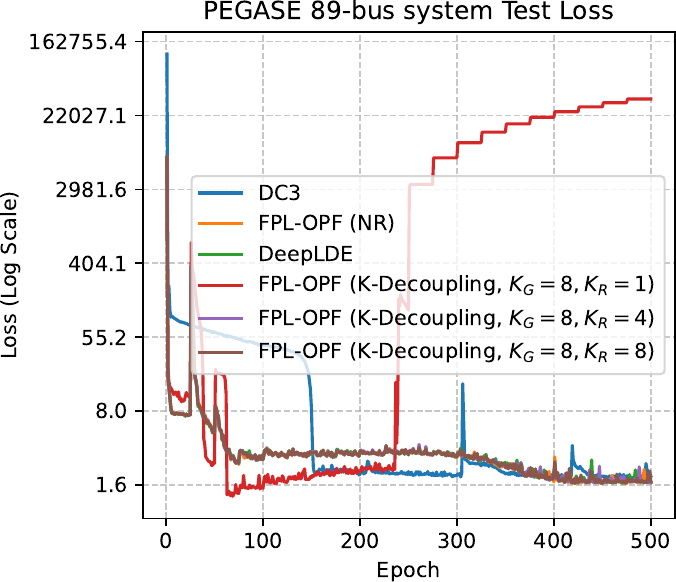}
				\caption{\footnotesize Test Loss convergence.  }
				\label{fig:fpl_ablation_loss}
			\end{subfigure}
			\begin{subfigure}[t]{0.22 \textwidth}
				\centering
				\includegraphics[width=  \textwidth]{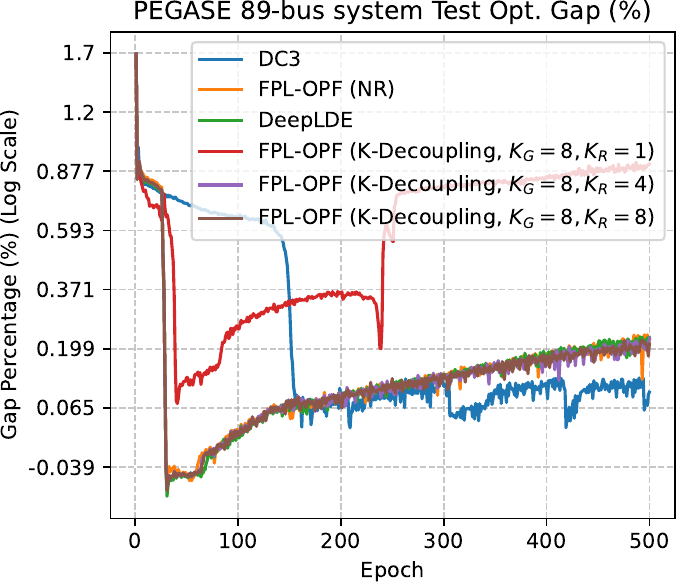}
				\caption{\footnotesize Test Optimality Gap (\%) convergence. }
				\label{fig:fpl_ablation_opt_gap}
			\end{subfigure}
			
			\caption{{Ablation study on the PEGASE 89-bus system about the number of refinement iterations $K_R$ in the FPL-OPF (K-Decoupling) method, compared with DC3, DeepLDE, and FPL-OPF (NR). The guide iterations are fixed at $K_G=8$.}}    
			\label{fig:fpl_ablation}
		\end{figure}
		
		\subsubsection{Constraint Satisfaction and Optimality}
		Table~\ref{tab:performance_comparison} summarizes feasibility and optimality on the four benchmarks, with the MIPS solver providing the reference optimum.
		The results first highlight the limitation of methods without an embedded physics layer. Although OPF-DNN attains artificially low objective values, its solutions are severely infeasible. 
		For example, its equality maximal mismatches are 32.56 on the PEGASE 89-bus system and 0.38 on the IEEE 118-bus system, which leads to misleading objective gaps of $-44.29\%$ and $-2.24\%$, respectively.
		\par 
		In contrast, DeepLDE and both FPL-OPF variants maintain very small equality mismatches across all systems, typically at or below the $10^{-8}$ level, and substantially reduce inequality violations relative to OPF-DNN and DC3.
		On the harder PEGASE 89-bus and NESTA 189-bus cases, FPL-OPF (K-Decoupling) achieves the best feasibility among the learned methods, with inequality mean mismatches of $1.11\times 10^{-4}$ and $5.21\times 10^{-7}$ and average inequality violation counts of 2.23 and 0.0075, respectively.
		\par 
		Under the fairer solver settings, physics-aware baselines become noticeably stronger on the first three benchmarks, and FPL-OPF remains competitive in optimality while retaining better constraint satisfaction on the large power grid.
		On the IEEE 57-bus system, FPL-OPF (NR) achieves the smallest learned optimality gap at 0.02\%.
		On the IEEE 118-bus system, DeepLDE and both FPL-OPF variants all attain a 0.06\% gap.
		On the large NESTA 189-bus system, FPL-OPF (K-Decoupling) clearly delivers the best trade-off, reducing the optimality gap to 1.48\%, compared with 4.30\% for FPL-OPF (NR) and 16.43\% for DeepLDE.
		These results show that FPL-OPF remains robust under the larger power grid, with its main advantage becoming especially pronounced on the more challenging system.
		% 针对时间部分还需要加上DeepLDE
		\subsubsection{Computational Efficiency}
		Figure~\ref{fig:time_comparision} shows that FPL-OPF preserves a clear computational advantage under the matched solver budgets.
		For training, FPL-OPF (K-Decoupling) requires only 0.716 seconds per epoch on the IEEE 118-bus system, compared with 32.774 seconds for DC3 and 3.228 seconds for DeepLDE, corresponding to speedups of about $45.8\times$ and $4.5\times$, respectively.
		The computational gain becomes even more remarkable on the large NESTA 189-bus system, where FPL-OPF (K-Decoupling) reduces the training time to 2.504 seconds per epoch, achieving a speedup of over $10\times$ compared to DeepLDE’s 26.212 seconds, versus 11.123 seconds for FPL-OPF (NR).
		This efficiency stems directly from our design that replaces the expensive full Jacobian factorization with the computationally efficient FDPF iterations, and we strictly maintain $O(n^2)$ complexity.
		For inference, once trained, FPL-OPF solves the full test set within a fraction of a second on the four systems. 
		In particular, FPL-OPF (K-Decoupling) requires 0.080 seconds on the IEEE 118-bus system and 0.322 seconds on the NESTA 189-bus system, whereas MIPS takes 58.402 seconds and 400.600 seconds, respectively, yielding speedups of about $730\times$ and $1244\times$. This substantial acceleration, together with the competitive feasibility and optimality reported above, highlights the practical value of FPL-OPF for time-sensitive AC-OPF applications. 
		Crucially, this strict $O(n^2)$ complexity allows the framework to scale to even larger grids without catastrophic memory overheads, as evidenced by our 500-bus system evaluations detailed in Appendix~\ref{sec:appendix_case500}.

		\subsection{Impact of Differentiable Refinement Iterations in FPL-OPF}
		\label{subsec:impact_kr}
		As discussed in \Cref{subsec:forward_pass,subsec:hybrid_solver}, the validity of the proposed implicit gradient approximation relies on the contraction property of the fixed-point operator. While a single FDPF iteration may not possess a sufficiently small contraction constant $\rho$ to guarantee strict gradient directional alignment in \Cref{thm:gradient_direction_align}, composing $K_R$ iterations yields a stronger block operator $T^{K_R}$ with $K_R$-dependent constants $\rho^{K_{R}}$, $L_{T,K_{R}}$, $L_{J,K_{R}}$, and $C_{1,K_{R}}$. We therefore conduct the training ablation and an empirical estimation of the constants of \Cref{thm:gradient_direction_align} on the PEGASE 89-bus system.
		\par
		We first perform the training ablation in Figure~\ref{fig:fpl_ablation}, where the number of non-differentiable guide iterations is fixed at $K_G=8$ and the differentiable refinement depth is varied across $K_R \in \{1,4,8\}$. Figures~\ref{fig:fpl_ablation_ineq_mean_mismatch} and \ref{fig:fpl_ablation_ineq_viol_num} show that the configuration with $K_R=1$ exhibits substantially larger inequality violations and more unstable training than the settings with larger $K_R$. The test loss and optimality gap in Figures~\ref{fig:fpl_ablation_loss} and \ref{fig:fpl_ablation_opt_gap} further show that $K_R=1$ fails to converge to a high-quality solution. In contrast, increasing $K_R$ to $4$ and $8$ drastically improves feasibility and training stability, indicating that differentiating through only one FDPF step does not provide a sufficiently accurate descent direction on this system.
					\begin{figure}[t]
			\centering
			\begin{subfigure}[t]{0.16\textwidth}
				\centering
				\includegraphics[width=\textwidth]{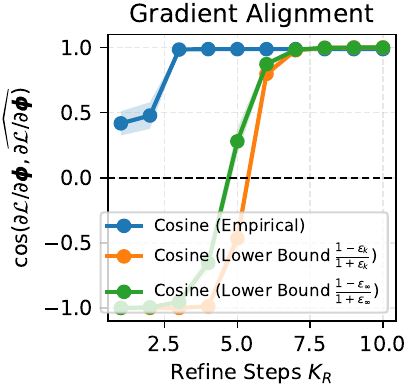}
				\caption{\footnotesize Gradient cosine similarity and theorem lower bounds.}
				\label{fig:theorem41_align}
			\end{subfigure}
			\begin{subfigure}[t]{0.15\textwidth}
				\centering
				\includegraphics[width=\textwidth]{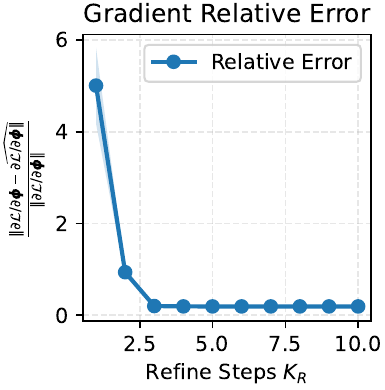}
				\caption{\footnotesize Relative gradient error.}
				\label{fig:theorem41_relerr}
			\end{subfigure}
			\begin{subfigure}[t]{0.15\textwidth}
				\centering
				\includegraphics[width=\textwidth]{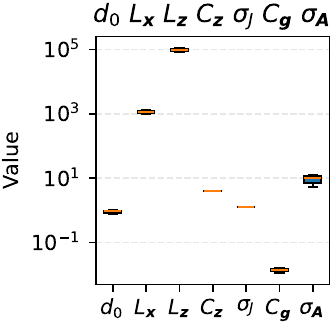}
				\caption{\footnotesize Fixed constants $d_0$, $L_{\bm{x}}$, $L_{\bm{z}}$, $C_{\bm{z}}$, $\sigma_{J}$, $C_{\bm{g}}$, and $\sigma_{\bm{A}}$.}
				\label{fig:theorem41_fixed}
			\end{subfigure}
			\begin{subfigure}[t]{0.16\textwidth}
				\centering
				\includegraphics[width=\textwidth]{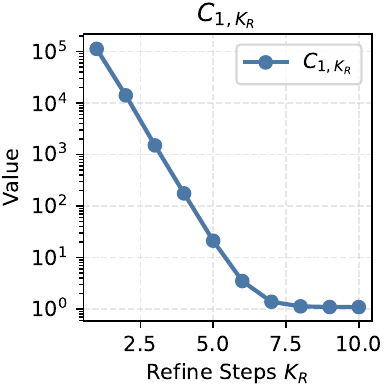}
				\caption{\footnotesize  Constant $C_{1,K_{R}}$.}
				\label{fig:theorem41_c1}
			\end{subfigure}
			\begin{subfigure}[t]{0.15\textwidth}
				\centering
				\includegraphics[width=\textwidth]{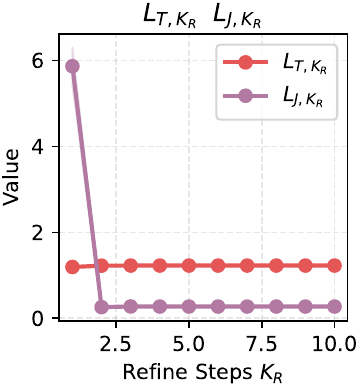}
				\caption{\footnotesize Composite Lipschitz constants $L_{T,K_{R}}$ and $L_{J,K_{R}}$.}
				\label{fig:theorem41_l}
			\end{subfigure}
			\begin{subfigure}[t]{0.15\textwidth}
				\centering
				\includegraphics[width=\textwidth]{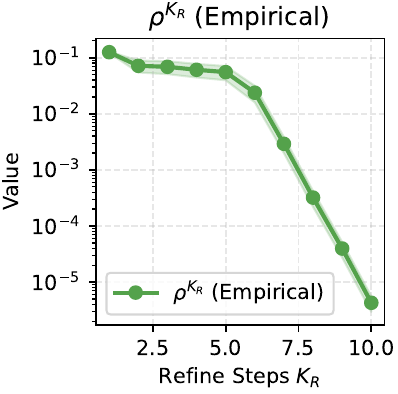}
				\caption{\footnotesize Contraction factors $\rho^{K_{R}}$ of the composite operator.}
				\label{fig:theorem41_rho}
			\end{subfigure}
			\caption{{Empirical estimation in \Cref{thm:gradient_direction_align} on the PEGASE 89-bus system. The solid curves represent the mean value, and the shaded region spans the standard deviation obtained over 50 samples.}}
			\label{fig:theorem41_empirical}
		\end{figure}
		To connect this behavior directly to \Cref{thm:gradient_direction_align}, we estimate the theorem constants on 50 samples drawn from the PEGASE 89-bus system by using the exact implicit gradient of DeepLDE as reference and treating the last $K_R$ FDPF refinement steps as the composite operator $T^{K_R}$.
		The key estimates are collected in Tables~\ref{tab:theorem41_constants} and \ref{tab:theorem41_composite_constants}, while Figure~\ref{fig:theorem41_empirical} visualizes the comprehensive and detailed constant estimation results from $K_R=1$ to $K_R=10$.
		Figure~\ref{fig:theorem41_fixed} shows that the fixed constants remain finite and concentrated in the neighborhood reached by the hybrid solver.
		This does not claim a global guarantee over arbitrary states, but it provides direct numerical evidence that the local boundedness and Lipschitz assumptions required by \Cref{thm:gradient_direction_align}  are satisfied in the actual power grid system on the operating region during training.
		\par

	\begin{table}[]
	\caption{{The fixed estimated constants  $d_0$, $L_{\bm{x}}$, $L_{\bm{z}}$, $C_{\bm{z}}$, $\sigma_{J}$, $C_{\bm{g}}$, and $\sigma_{\bm{A}}$ of \Cref{thm:gradient_direction_align} on the PEGASE 89-bus system. All entries report mean $\pm$ standard deviation.}}
	\label{tab:theorem41_constants}
	\resizebox{0.8\linewidth}{!}{
	\begin{tabular}{@{}cl|cl@{}}
		\toprule
		\multicolumn{1}{c}{Constant} & Mean $\pm$ Std.                & \multicolumn{1}{c}{Constant} & Mean $\pm$ Std.                   \\ \midrule
		$C_{\bm{z}}$                             & $3.982 \pm 0.027$              & $d_0$                             & $0.918 \pm 0.076$                 \\
		$\sigma_{J}$                             & $1.234 \pm 0.005$              & $C_{\bm{g}}$                      & $(1.351 \pm 0.145)\times 10^{-2}$ \\
		$L_{\bm{x}}$                             & $(1.150 \pm 0.085)\times 10^3$ & $\sigma_{\bm{A}}$                 & $9.61 \pm 2.25$                   \\
		$L_{\bm{z}}$                             & $(9.818 \pm 0.730)\times 10^4$ &                                   &                                   \\ \bottomrule
	\end{tabular}}
\end{table}

% \usepackage{booktabs}
\begin{table}[]
	\caption{{The estimated composite-operator constants  $C_{1,K_{R}}$, $\rho^{K_{R}}$, $L_{T,K_{R}}$, and $L_{J,K_{R}}$ of \Cref{thm:gradient_direction_align} on the PEGASE 89-bus system. All entries report mean $\pm$ standard deviation.}}
	\label{tab:theorem41_composite_constants}
	\resizebox{0.95\linewidth}{!}{
	\begin{tabular}{@{}lllll@{}}
		\toprule
		$K_R$ & $\rho^{K_{R}}$ (Empirical)              & $L_{T,K_{R}}$         & $L_{J,K_{R}}$         & $C_{1,K_{R}}$                      \\ \midrule
		$1$   & $0.125 \pm 0.004$                 & $1.195 \pm 0.001$ & $5.870 \pm 0.442$ & $(1.122 \pm 0.084)\times 10^5$ \\
		$4$   & $0.0609 \pm 0.0166$               & $1.228 \pm 0.005$ & $0.274 \pm 0.008$ & $(1.755 \pm 0.109)\times 10^2$ \\
		$8$   & $(3.211 \pm 1.023)\times 10^{-4}$ & $1.228 \pm 0.005$ & $0.274 \pm 0.008$ & $1.127 \pm 0.029$              \\ \bottomrule
	\end{tabular} 
}
\end{table}
		\par
		Figures~\ref{fig:theorem41_c1}--\ref{fig:theorem41_rho} together with Tables~\ref{tab:theorem41_constants}--\ref{tab:theorem41_composite_constants} further show how composing FDPF refinement steps improves the bound. From $K_R=1$ to $K_R=4$ and $8$,  the  empirical contraction factor shrink rapidly, and the constant $C_{1,K_{R}}$ collapses by several orders of magnitude.
		In comparison, $L_{T,K_{R}}$ changes little once $K_R \ge 2$, while $L_{J,K_{R}}$ drops sharply then stabilizes. Hence, the dominant improvement in the theorem lower bound comes from the geometric decay of $\rho^{K_{R}}$ and the associated reduction of the transient term.
		\par
		This improvement is reflected directly in the gradients. In Figure~\ref{fig:theorem41_align}, the empirical cosine similarity between the exact and approximate parameter gradients increases from $0.42$ at $K_R=1$ and $0.48$ at $K_R=2$ to above $0.98$ once $K_R \ge 3$, while Figure~\ref{fig:theorem41_relerr} shows that the relative gradient error drops from about $5.0$ to about $0.19$. The theorem-based lower bounds are conservative but informative.
		The asymptotic lower bound $\frac{1-\varepsilon_{\infty}}{1+\varepsilon_{\infty}}$ becomes positive at $K_R=5$, and the lower bound $\frac{1-\varepsilon_{k}}{1+\varepsilon_{k}}$ becomes positive at $K_R \ge 6$. This explains why the practical transition in Figure~\ref{fig:fpl_ablation} already occurs between $K_R=1$ and $K_R \in \{4,8\}$, while a larger $K_R$ provides an additional theory-driven safety margin. In practice, we therefore choose the smallest refinement depth at which the empirical alignment has essentially saturated and the forward cost remains low.
		For the PEGASE 89-bus case, this transition occurs around $K_R=4$.
		
		\section{Conclusion}
		\label{sec:conclusion}
		We introduced FPL-OPF, a physics-aware unsupervised framework that solves AC-OPF by embedding an FDPF solver as an implicit layer. Our method ensures strict adherence to the physical laws of power systems while circumventing the significant computational overhead of traditional approaches. We theoretically proved that the gradient derived from the refinement stage aligns well with the true implicit gradient, enabling efficient automatic differentiation without Jacobian inversions. 
		Experiments demonstrate that FPL-OPF offers a substantial acceleration compared to existing methods without compromising physical feasibility or optimality.
		{Future work will focus on developing more sophisticated theoretical estimators to strictly determine $K_R$, extending the framework to the dynamic grid topology, and designing distributed learning architectures for large-scale, multi-area power systems.}
%		\begin{acks}
		\section*{Acknowledgments}
		This work was supported by National Natural Science Foundation of China (62303319), ShanghaiTech AI4S Initiative SHTAI4S202404, HPC Platform of ShanghaiTech University, and MoE Key Laboratory of Intelligent Perception and Human-Machine Collaboration (ShanghaiTech University).
%		\end{acks}
		\clearpage
		\bibliographystyle{ACM-Reference-Format}
		\bibliography{intro}
		
		\clearpage
		
		\appendix
		\onecolumn
		\setcounter{theorem}{0}
		\setcounter{equation}{0}%将公式编号归0
		\renewcommand{\theequation}{A.\arabic{equation}}
		
		\setcounter{table}{0}
		\renewcommand{\thetable}{\textsc{A}.\arabic{table}}
		
		\setcounter{figure}{0}
		\renewcommand{\thefigure}{A.\arabic{figure}}
		\section*{Appendix}

		\section{Proof for \Cref{thm:gradient_direction_align}}
		\begin{theorem}[Gradient  Directional Alignment]
			\label{app:proof_gradient_direction_align}
			Assume Lemmas \ref{lemma:fdpf_contraction} and \ref{lemma:os_id_error} and Assumptions \ref{ass:lipschitz_loss_grad}--\ref{ass:op_norm_bound} hold.
			Let $d_0 \triangleq \|\bm{z}_0 - \bm{z}^{\star}\|$, $C_1 \triangleq \rho (L_{\bm{x}} + L_{\bm{z}} \sigma_{J}) + C_{\bm{z}} L_J$, and 
			$\varepsilon_{k}\triangleq\frac{\sigma_{\bm{A}}}{C_{\bm{g}}} \left(  \frac{\rho L_T C_{\bm{z}} }{1-\rho} 	+ \rho^k d_0 C_1 \right)$.
			Then, for every $k$ satisfying $\varepsilon_{k} < 1$, the cosine similarity between the true gradient $\frac{\partial \mathcal{L}}{\partial \bm{\phi}}$ and its approximation gradient $\widehat{\frac{\partial \mathcal{L}}{\partial \bm{\phi}} }$ satisfies:
			$$
			\left \langle 
			\frac{\partial \mathcal{L}}{\partial \bm{\phi}}, 
			\widehat{\frac{\partial \mathcal{L}}{\partial \bm{\phi}}} 
			\right \rangle
			\ge
			\frac{1-\varepsilon_{k}}{1+\varepsilon_{k}}
			> 0.
			$$
		\end{theorem}
		
		\begin{proof}
			For notation brevity, define
			$
			\bm{a}^{\star}
			\triangleq
			\partial_{\bm{x}} \mathcal{L}_{\mathrm{part}}(\bm{x},\bm{z}^{\star}),
			\;
			\bm{b}^{\star}
			\triangleq
			\partial_{\bm{z}} \mathcal{L}_{\mathrm{total}}(\bm{x},\bm{z}^{\star}),
			\;
			\bm{J}^{\star}
			\triangleq
			\bm{J}^{\mathrm{ID}}_{\bm{x}}\bm{z}^{\star},
			$
			and $\bm{a}^{k+1}, \; \bm{b}^{k+1}, \; \bm{J}^{k+1}$ denote their counterparts evaluated at the $(k+1)$-th iterate.
			\newline
			Let
			\begin{equation}
				\bm{v}^{\star} \triangleq \bm{a}^{\star} + \bm{b}^{\star}\bm{J}^{\star},
				\quad
				\bm{v}^{k+1} \triangleq \bm{a}^{k+1} + \bm{b}^{k+1}\bm{J}^{k+1},
				\quad
				\Delta \bm{v} \triangleq \bm{v}^{k+1} - \bm{v}^{\star}.
			\end{equation}
			Then,
			\begin{equation}
				\frac{\partial \mathcal{L}}{\partial \bm{\phi}} 
				=\bm{g}^{\star} = \bm{v}^{\star}\bm{A},
				\quad
				\widehat{\frac{\partial \mathcal{L}}{\partial \bm{\phi}}} =\widehat{\bm{g}}^{k+1} = \bm{v}^{k+1}\bm{A}
				= \bm{g}^{\star} + \Delta \bm{v}\,\bm{A}.
			\end{equation}
			We first bound $\Delta \bm{v}$. By the triangle inequality,
			\begin{equation}
				\|\Delta \bm{v}\|
				\le
				\|\bm{a}^{k+1} - \bm{a}^{\star}\|
				+
				\|\bm{b}^{k+1}\| \, \|\bm{J}^{k+1} - \bm{J}^{\star}\|_{\mathrm{op}}
				+
				\|\bm{b}^{k+1} - \bm{b}^{\star}\| \, \|\bm{J}^{\star}\|_{\mathrm{op}}.
			\end{equation}
			By smoothness and contraction,
			\begin{equation}
				\begin{aligned}
					\|\bm{a}^{k+1} - \bm{a}^{\star}\| 
					\leq L_{\bm{x}} \rho^{k+1} d_0, 
					\quad  
					\|\bm{b}^{k+1} - \bm{b}^{\star}\| \leq L_{\bm{z}} \rho^{k+1} d_0, 
					\quad
					\|\bm{J}^{k+1} - \bm{J}^{\star}\|_{\mathrm{op}} 
					\leq 
					\frac{\rho L_T}{1-\rho} + L_J \rho^{k} d_0. 
				\end{aligned}
			\end{equation}
			Combining with $\|\bm{b}^{k+1}\| \leq C_{\bm{z}}$ and $\|\bm{J}^{\star}\| \leq \sigma_{J}$,
			\begin{equation}
				\|\Delta \bm{v}\|
				\le
				\frac{\rho L_T C_{\bm{z}} }{1-\rho}
				+
				\rho^k d_0 \left[ \rho (L_{\bm{x}} + L_{\bm{z}} \sigma_{J}) + C_{\bm{z}} L_J \right]
				=
				\frac{\rho L_T C_{\bm{z}} }{1-\rho}
				+
				\rho^k d_0 C_1 .
			\end{equation}
			By Assumption~\ref{ass:op_norm_bound},
			\begin{equation}
							\|\Delta \bm{v}\,\bm{A}\|
				\le
				\|\Delta \bm{v}\| \|\bm{A}\|_{\mathrm{op}}
				\le
				\sigma_{\bm{A}} \|\Delta \bm{v}\| .
			\end{equation}
			Next, by Assumption~\ref{ass:non_vanish_grad},
			\begin{equation}
							\|\bm{g}^{\star}\|
				=
				\|\bm{v}^{\star}\bm{A}\|
				\ge C_{\bm{g}} .
			\end{equation}
			Hence,
			\begin{equation}
				\begin{aligned}
					\langle \bm{g}^{\star}, \widehat{\bm{g}}^{k+1} \rangle
					=
					\langle \bm{g}^{\star}, \bm{g}^{\star} + \Delta \bm{v}\,\bm{A} \rangle
					\ge 
					\|\bm{g}^{\star}\|^2
					-
					\|\bm{g}^{\star}\|  \|\Delta \bm{v}\,\bm{A}\| 
					\ge
					\|\bm{g}^{\star}\|
					\left(
					\|\bm{g}^{\star}\| - \sigma_{\bm{A}} \|\Delta \bm{v}\|
					\right).
				\end{aligned}
			\end{equation}
			Similarly,
			\begin{equation}
				\|\widehat{\bm{g}}^{k+1}\|
				\le
				\|\bm{g}^{\star}\| + \|\Delta \bm{v}\,\bm{A}\|
				\le
				\|\bm{g}^{\star}\| + \sigma_{\bm{A}} \|\Delta \bm{v} \| .
			\end{equation}
			Therefore,
			\begin{equation}
				\cos \bigl(\bm{g}^{\star}, \widehat{\bm{g}}^{k+1}\bigr)
				=
				\frac{\langle \bm{g}^{\star}, \widehat{\bm{g}}^{k+1} \rangle}
				{\|\bm{g}^{\star}\|  \|\widehat{\bm{g}}^{k+1}\|}
				\ge
				\frac{\|\bm{g}^{\star}\| - \sigma_{\bm{A}} \|\Delta \bm{v}\|}
				{\|\bm{g}^{\star}\| + \sigma_{\bm{A}} \|\Delta \bm{v}\|}.
			\end{equation}
			Using $\|\bm{g}^{\star}\| \ge C_{\bm{g}}$ and the bound on $\|\Delta \bm{v}\|$, we obtain
			\begin{equation}
				\cos \left (\bm{g}^{\star}, \widehat{\bm{g}}^{k+1} \right)
				\ge
				\frac{
					1 - \frac{\sigma_{\bm{A}}}{C_{\bm{g}}}\left(  \frac{\rho L_T C_{\bm{z}} }{1-\rho} + \rho^k d_0 C_1 \right)
				}{
					1 + \frac{\sigma_{\bm{A}}}{C_{\bm{g}}}\left(  \frac{\rho L_T C_{\bm{z}} }{1-\rho} + \rho^k d_0 C_1 \right)
				}
				=
				\frac{1-\varepsilon_{k}}{1+\varepsilon_{k}}.
			\end{equation}
			Thus, whenever $\varepsilon_{k} < 1$, the cosine similarity is strictly positive.
		\end{proof}

		% figure ineq mean mismatch
		\begin{figure*}[ht]
			\centering
			\begin{subfigure}[t]{0.22 \textwidth}
				\centering
				\includegraphics[width= \textwidth]{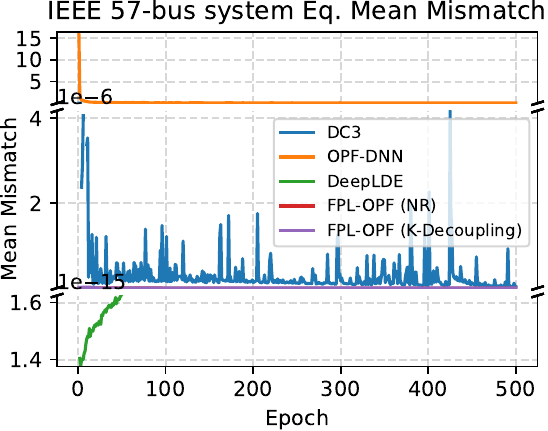}
				\caption{\footnotesize Equality mean mismatch comparison at the IEEE 57-bus system.}
				\label{fig:appedix_case57_eq_mean_mismatch}
				
			\end{subfigure}
			\begin{subfigure}[t]{0.22 \textwidth}
				\centering
				\includegraphics[width= \textwidth]{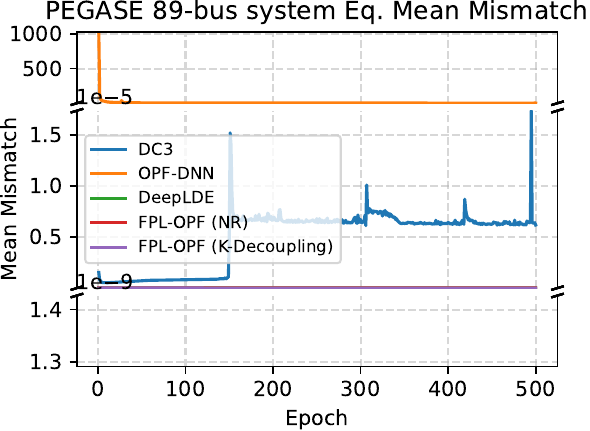}
				\caption{\footnotesize Equality mean mismatch comparison at PEGASE 89-bus system.}
				
				\label{fig:appedix_case89_eq_mean_mismatch}
			\end{subfigure}
			\begin{subfigure}[t]{0.22 \textwidth}
				\centering
				\includegraphics[width= \textwidth]{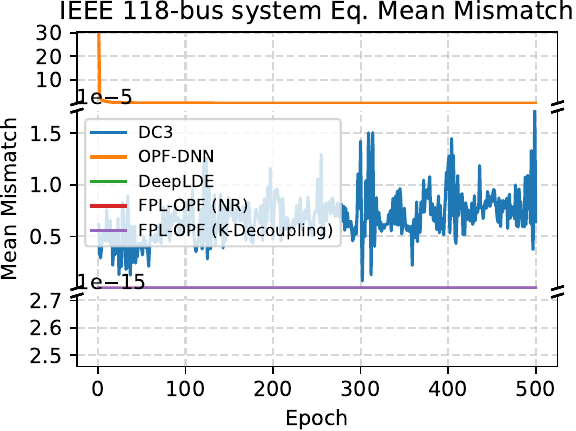}
				\caption{\footnotesize Equality mean mismatch comparison at the IEEE 118-bus system.}
				\label{fig:appedix_case118_eq_mean_mismatch}
			\end{subfigure}
						\begin{subfigure}[t]{0.22 \textwidth}
				\centering
				\includegraphics[width= \textwidth]{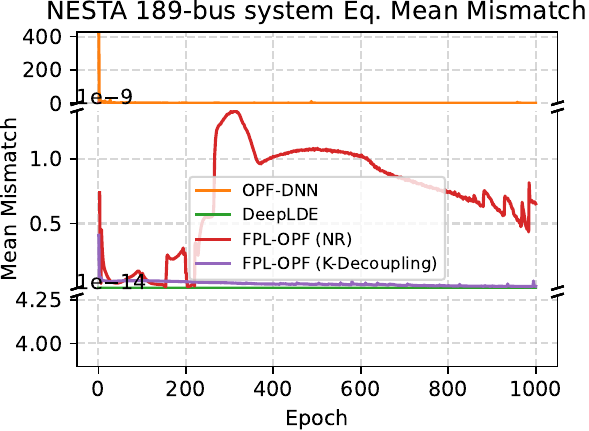}
				\caption{\footnotesize Equality mean mismatch comparison at the NESTA 189-bus system.}
				\label{fig:appedix_case189_eq_mean_mismatch}
			\end{subfigure}
			\caption{Equality mean mismatch comparison on (a) IEEE 57-bus, (b) PEGASE 89-bus, (c) IEEE 118-bus, and (d) NESTA 189-bus systems.}    
			\label{fig:appendix_eq_mean_mismatch}
		\end{figure*}
		
		% figure ineq mean mismatch
		\begin{figure*}[ht]
			\centering
			\begin{subfigure}[t]{0.23 \textwidth}
				\centering
				\includegraphics[width= \textwidth]{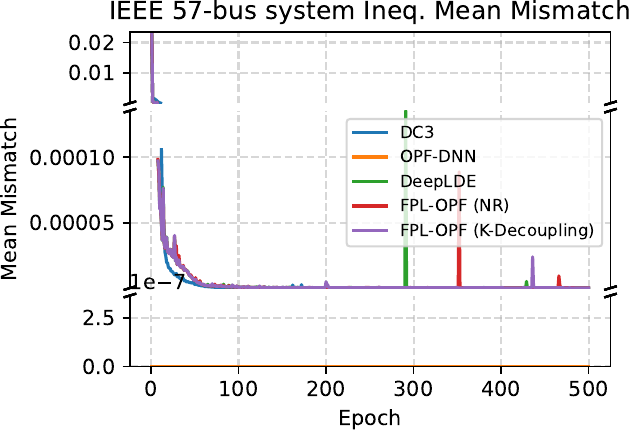}
				\caption{\footnotesize Inequality mean mismatch comparison at the IEEE 57-bus system.}
				\label{fig:appedix_case57_ineq_mean_mismatch}
			\end{subfigure}
			\begin{subfigure}[t]{0.23 \textwidth}
				\centering
				\includegraphics[width= \textwidth]{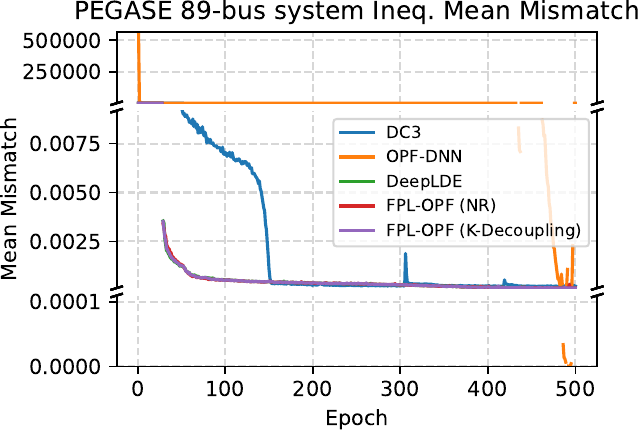}
				\caption{\footnotesize Inequality mean mismatch comparison at PEGASE 89-bus system.}
				
				\label{fig:appedix_case89_ineq_mean_mismatch}
			\end{subfigure}
			\begin{subfigure}[t]{0.21 \textwidth}
				\centering
				\includegraphics[width= \textwidth]{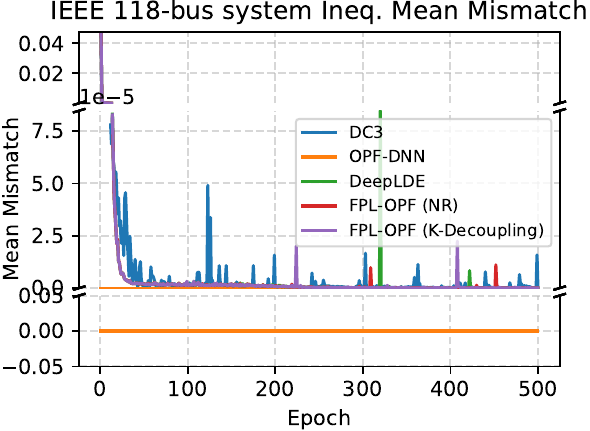}
				\caption{\footnotesize Inequality mean mismatch comparison at the IEEE 118-bus system.}
				\label{fig:appedix_case118_ineq_mean_mismatch}
			\end{subfigure}
			\begin{subfigure}[t]{0.21 \textwidth}
				\centering
				\includegraphics[width= \textwidth]{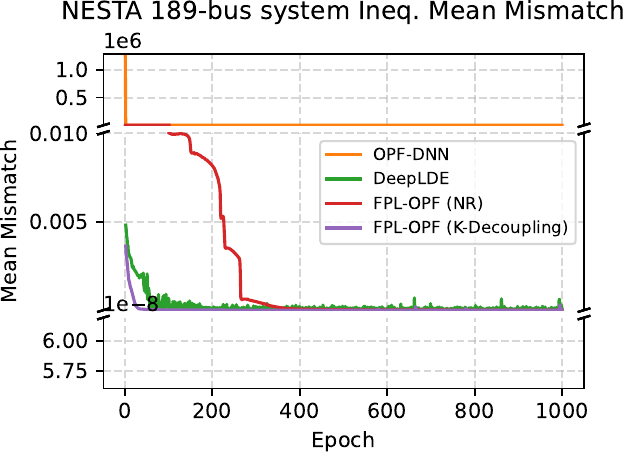}
				\caption{\footnotesize Inequality mean mismatch comparison at the NESTA 189-bus system.}
				\label{fig:appedix_case189_ineq_mean_mismatch}
			\end{subfigure}
			\caption{Inequality mean mismatch comparison on (a) IEEE 57-bus, (b) PEGASE 89-bus, (c) IEEE 118-bus, and (d) NESTA 189-bus systems.}    
			\label{fig:appendix_ineq_mean_mismatch}
		\end{figure*}
		
		% figure ineq viol num
		\begin{figure*}[ht]    
			\centering
			
			\begin{subfigure}[t]{0.22 \textwidth}
				\centering
				\includegraphics[width= \textwidth]{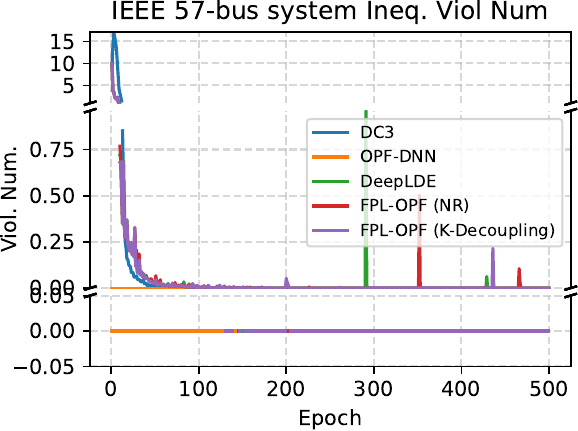}
				\caption{\footnotesize Inequality violation number comparison at IEEE 57-bus system.}
				\label{fig:appedix_case57_ineq_viol_num}
				
			\end{subfigure}
			\begin{subfigure}[t]{0.22 \textwidth}
				\centering
				\includegraphics[width= \textwidth]{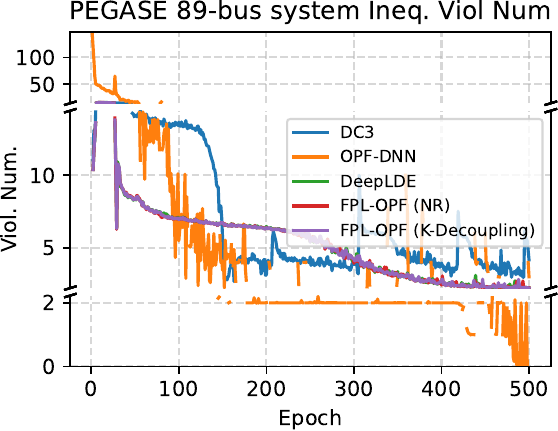}
				\caption{\footnotesize Inequality violation number comparison at PEGASE 89-bus system.}
				
				\label{fig:appedix_case89_ineq_viol_num}
			\end{subfigure}
			\begin{subfigure}[t]{0.22 \textwidth}
				\centering
				\includegraphics[width= \textwidth]{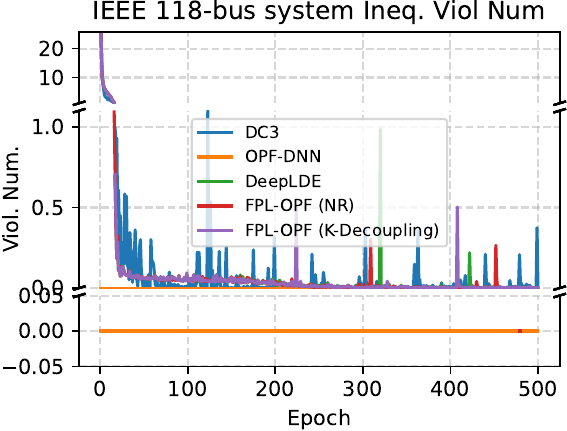}
				\caption{\footnotesize Inequality violation number comparison at IEEE 118-bus system.}
				\label{fig:appedix_case118_ineq_viol_num}
			\end{subfigure}
						\begin{subfigure}[t]{0.22 \textwidth}
				\centering
				\includegraphics[width= \textwidth]{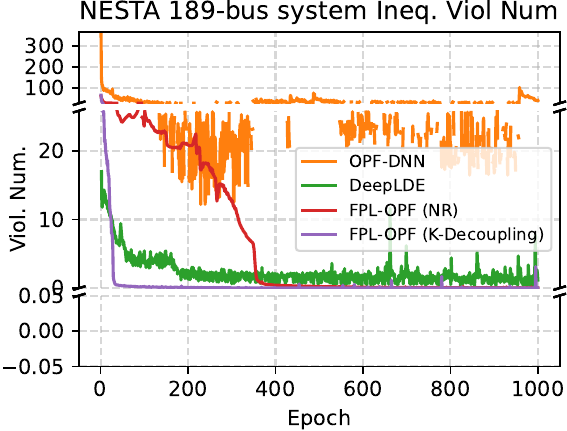}
				\caption{\footnotesize Inequality violation number comparison at NESTA 189-bus system.}
				\label{fig:appedix_case189_ineq_viol_num}
			\end{subfigure}
			\caption{Inequality violation number comparison on (a) IEEE 57-bus, (b) PEGASE 89-bus, (c) IEEE 118-bus, and (d) NESTA 189-bus systems.}    
			\label{fig:appendix_ineq_viol_num}
		\end{figure*}
		
		% figure test loss
		\begin{figure*}[ht]    
			\centering
			
			\begin{subfigure}[t]{0.22 \textwidth}
				\centering
				\includegraphics[width= \textwidth]{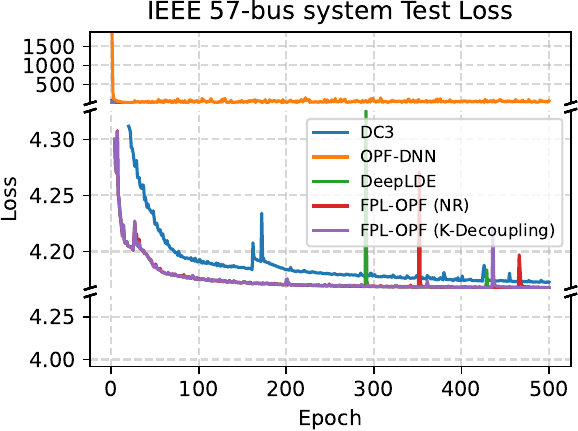}
				\caption{\footnotesize Loss comparison at IEEE 57-bus system.}
				\label{fig:appedix_case57_loss_curve}
				
			\end{subfigure}
			\begin{subfigure}[t]{0.22 \textwidth}
				\centering
				\includegraphics[width= \textwidth]{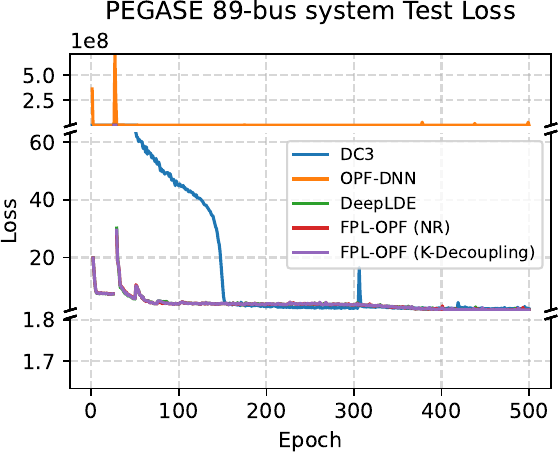}
				\caption{\footnotesize Loss comparison at PEGASE 89-bus system.}
				
				\label{fig:appedix_case89_loss_curve}
			\end{subfigure}
			\begin{subfigure}[t]{0.22 \textwidth}
				\centering
				\includegraphics[width= \textwidth]{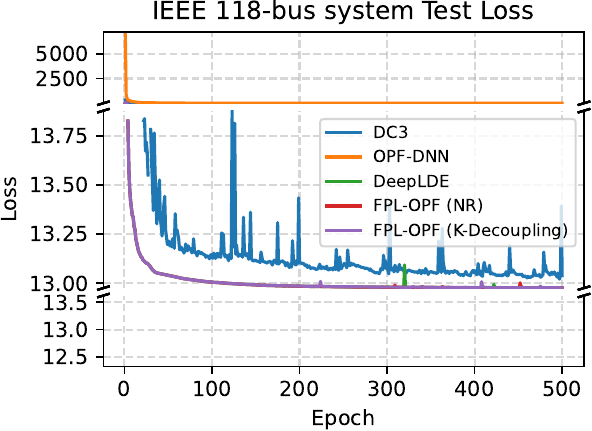}
				\caption{\footnotesize Loss comparison at IEEE 118-bus system.}
				\label{fig:appedix_case118_loss_curve}
			\end{subfigure}
						\begin{subfigure}[t]{0.22 \textwidth}
				\centering
				\includegraphics[width= \textwidth]{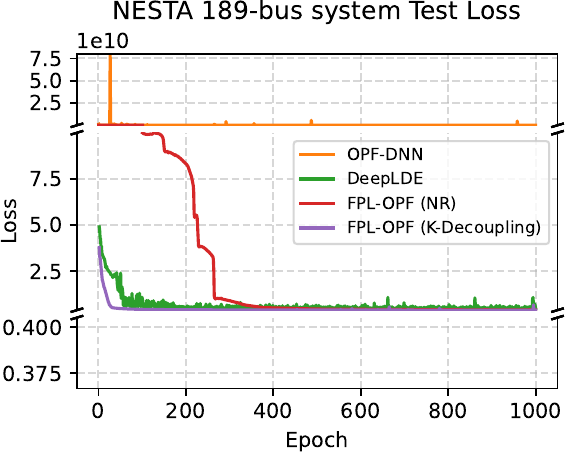}
				\caption{\footnotesize Loss comparison at NESTA 189-bus system.}
				\label{fig:appedix_case189_loss_curve}
			\end{subfigure}
			\caption{Loss comparison on (a) IEEE 57-bus, (b) PEGASE 89-bus, (c) IEEE 118-bus, and (d) NESTA 189-bus systems.}    
			\label{fig:appendix_loss_curve}
		\end{figure*}
		
		% figure test obj
		\begin{figure*}[ht]    
			\centering
			
			\begin{subfigure}[t]{0.22 \textwidth}
				\centering
				\includegraphics[width= \textwidth]{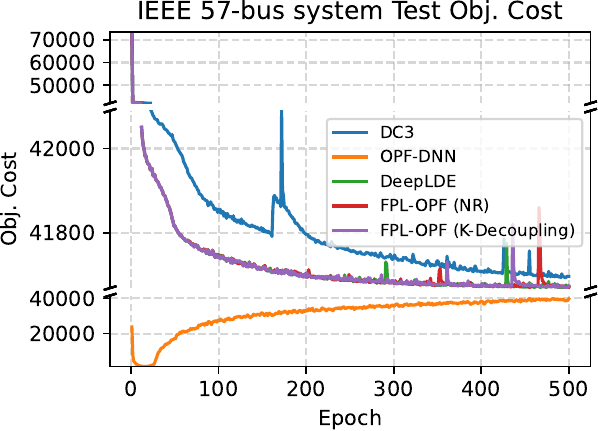}
				\caption{\footnotesize Objective cost comparison at IEEE 57-bus system.}
				\label{fig:appendix_case57_obj_cost}
				
			\end{subfigure}
			\begin{subfigure}[t]{0.22 \textwidth}
				\centering
				\includegraphics[width= \textwidth]{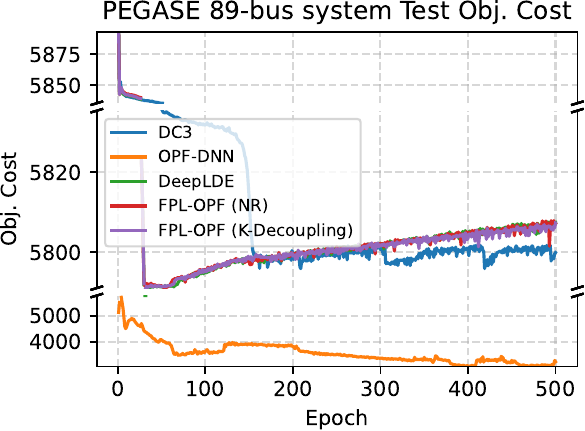}
				\caption{\footnotesize Objective cost comparison at PEGASE 89-bus system.}
				
				\label{fig:appendix_case89_obj_cost}
			\end{subfigure}
			\begin{subfigure}[t]{0.22 \textwidth}
				\centering
				\includegraphics[width= \textwidth]{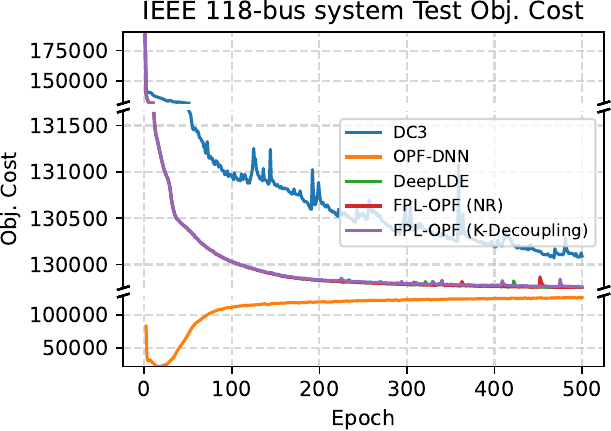}
				\caption{\footnotesize Objective cost comparison at IEEE 118-bus system.}
				\label{fig:appendix_case118_obj_cost}
			\end{subfigure}
						\begin{subfigure}[t]{0.22 \textwidth}
				\centering
				\includegraphics[width= \textwidth]{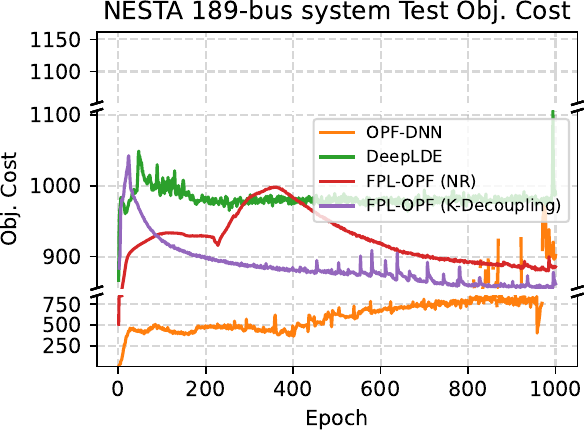}
				\caption{\footnotesize Objective cost comparison at NESTA 189-bus system.}
				\label{fig:appendix_case189_obj_cost}
			\end{subfigure}
			\caption{Objective cost comparison on (a) IEEE 57-bus, (b) PEGASE 89-bus, (c) IEEE 118-bus, and (d) NESTA 189-bus systems.}    
			\label{fig:appendix_obj_cost}
		\end{figure*}
		
\section{Supplementary Experimental Results for PGLib 500-bus System}
\label{sec:appendix_case500}
To further evaluate the scalability and robustness of the proposed framework, we conduct experiments on the PGLib 500-bus system\footnote{\url{https://github.com/power-grid-lib/pglib-opf/blob/master/pglib_opf_case500_goc.m}}, a significantly larger and more complex benchmark comprising 500 buses, 597 branches, and 56 generators. This system introduces a highly constrained optimization space defined by 1,000 equality constraints and 2,418 inequality constraints.
\par
The dataset generation, load profile perturbations, and general training configurations remain consistent with the NESTA 189-bus system setup. However, to accommodate the increased system scale, we expand the neural network architecture to use hidden layers of 6000-6000, reduce the learning rate to 1e-5, and set the outer iterations to 80, resulting in a total of 2,000 training epochs.
\par
In this evaluation, we compare the performance of FPL-OPF (K-Decoupling) exclusively against the traditional MIPS solver. Physics-aware learning baselines relying on exact implicit differentiation or Newton-Raphson steps, such as DC3, DeepLDE, and FPL-OPF (NR), are excluded from this benchmark. The dense matrix operations and large-scale Jacobian inversions required by these methods, combined with the memory overhead of backpropagation, exceed the 24 GB capacity of an NVIDIA RTX 4090D GPU, even at a severely reduced batch size of 16. These Out-Of-Memory failures highlight a critical scalability bottleneck in existing implicit-layer-based approaches.
\par
Conversely, our proposed FPL-OPF (K-Decoupling) successfully scales to this challenging scenario by avoiding Jacobian matrix formations and strictly maintaining an $O(n^2)$ memory and time complexity. As detailed in Table~\ref{tab:appendix_case500_perf} and Figures~\ref{fig:appendix_case500_curves} and \ref{fig:appendix_time_comparision_500}, FPL-OPF (K-Decoupling) trains stably and achieves a highly competitive optimality gap of 0.35\% while strictly satisfying the complex physical constraints. Furthermore, it demonstrates an overwhelming computational advantage, accelerating the inference process by a factor of over $1100 \times$ compared to the MIPS solver, i.e., 0.746 seconds versus 827.748 seconds. These results underscore the superior memory efficiency, computational speed, and practical viability of the K-Decoupling strategy for large-scale power grid operations.
	\begin{table*}[!h]
		\centering
		\caption{Performance Comparison On PGLib 500-bus System}
		\label{tab:appendix_case500_perf}
		\resizebox{0.85\linewidth}{!}{
		\begin{tabular}{lcccccccc}
			\toprule
			Method & \makecell[c]{Eq. Mean\\Mismatch} & \makecell[c]{Eq. Max.\\Mismatch} & \makecell[c]{Eq. Viol.\\Num.} & \makecell[c]{Ineq. Mean\\Mismatch} & \makecell[c]{Ineq. Max\\Mismatch} & \makecell[c]{Ineq. Viol.\\Num.} & \makecell[c]{Objective\\Cost (\$)} & \makecell[c]{Objective\\Gap (\%)} \\
			\midrule
			\multicolumn{9}{c}{\textbf{PGLib 500-bus system}} \\
			\midrule
			MIPS Solver & 0.00e+00 & 0.00e+00 & 0.00e+00 & 0.00e+00 & 0.00e+00 & 0.00e+00 & 72566.66 & 0.00\% \\
			\cline{2-9}
			FPL-OPF (K-Decoupling) & \textbf{1.07e-13} & \textbf{3.87e-12} & \textbf{0.00e+00} & \textbf{1.54e-05} & \textbf{0.0302} & \textbf{0.6550} & \textbf{72821.54} & \textbf{0.35\%} \\
			\bottomrule
		\end{tabular}
	}
	\end{table*}
	
			% figure case 500
	\begin{figure*}[ht]    
		\centering
				\begin{subfigure}[t]{0.22 \textwidth}
			\centering
			\includegraphics[width= \textwidth]{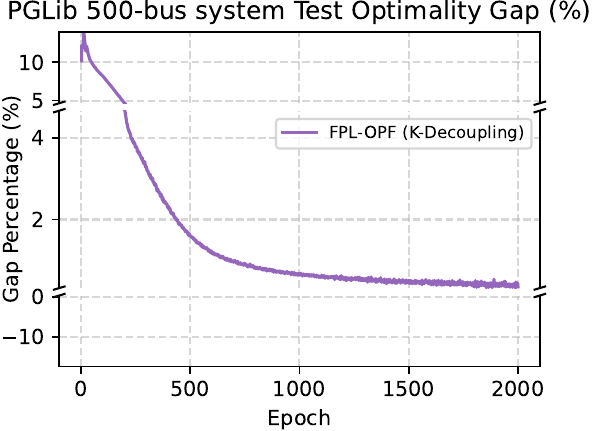}
			\caption{\footnotesize Test Optimality gap at PGLib 500-bus system.}
			
			\label{fig:appedix_case500_opt_gap_case500}
		\end{subfigure}
		\begin{subfigure}[t]{0.22 \textwidth}
			\centering
			\includegraphics[width= \textwidth]{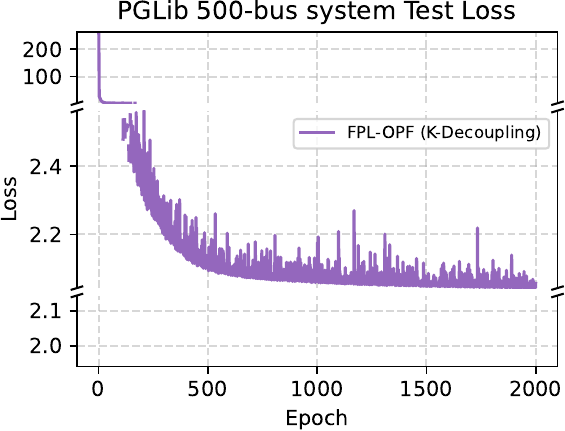}
			\caption{\footnotesize Test loss curve at PGLib 500-bus system.}
			\label{fig:appedix_case500_loss_curve}
			
		\end{subfigure}
		\begin{subfigure}[t]{0.22 \textwidth}
			\centering
			\includegraphics[width= \textwidth]{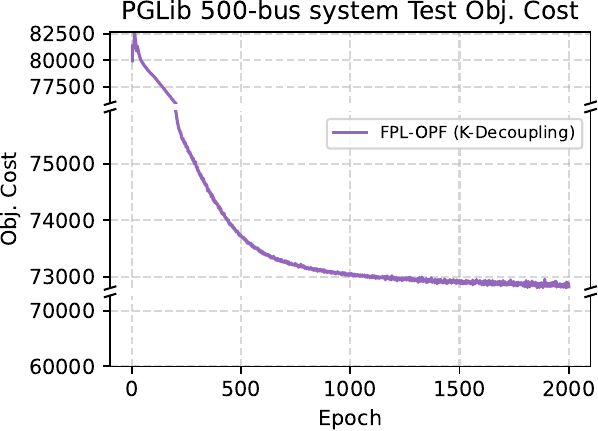}
			\caption{\footnotesize Test objective at PGLib 500-bus system.}
			
			\label{fig:appedix_case500_obj_cost}
		\end{subfigure}
		\\
		\begin{subfigure}[t]{0.22 \textwidth}
			\centering
			\includegraphics[width= \textwidth]{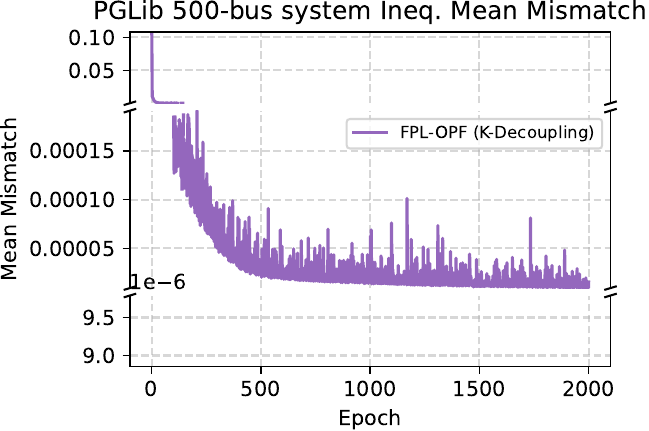}
			\caption{\footnotesize Inequality mean mismatch at PGLib 500-bus system.}
			\label{fig:appedix_case500_ineq_mismatch}
		\end{subfigure}
		\begin{subfigure}[t]{0.22 \textwidth}
			\centering
			\includegraphics[width= \textwidth]{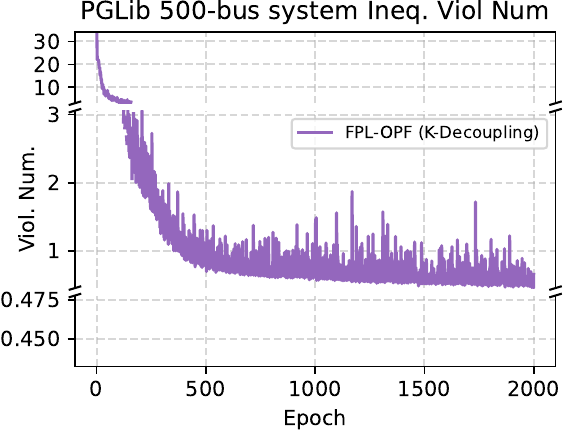}
			\caption{\footnotesize Inequality violation number at PGLib 500-bus system.}
			\label{fig:appedix_case500_ineq_viol_num}
		\end{subfigure}
		\begin{subfigure}[t]{0.22 \textwidth}
			\centering
			\includegraphics[width= \textwidth]{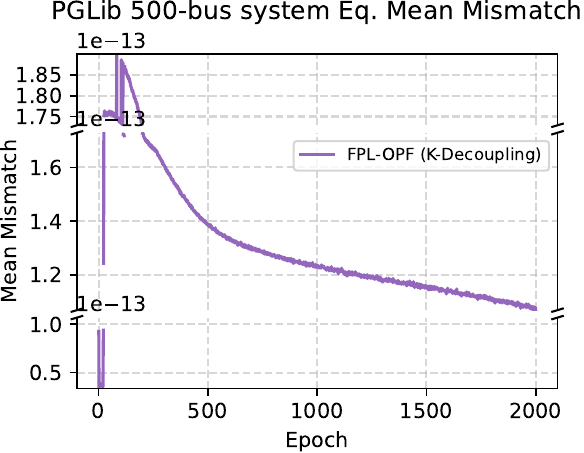}
			\caption{\footnotesize Equality mean mismatch at PGLib 500-bus system.}
			\label{fig:appedix_case500_eq_mismatch}
		\end{subfigure}
		\begin{subfigure}[t]{0.22 \textwidth}
			\centering
			\includegraphics[width= \textwidth]{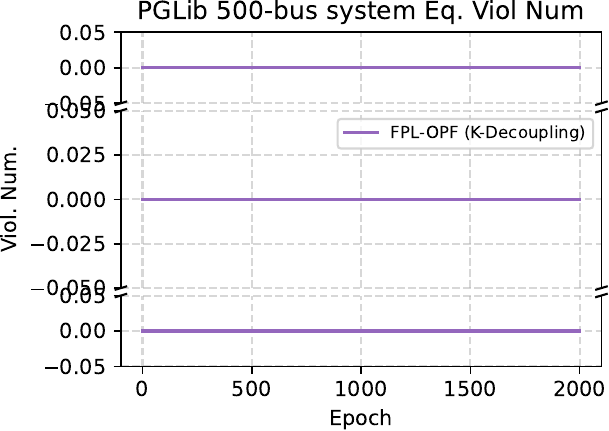}
			\caption{\footnotesize Equality violation number at PGLib 500-bus system.}
			\label{fig:appedix_case500_eq_viol_num}
		\end{subfigure}
		\caption{PGLib 500-bus system training performance visualization.}    
		\label{fig:appendix_case500_curves}
	\end{figure*}
	
	\begin{figure*}[ht]    
		\centering
		\includegraphics[width= 0.3\textwidth]{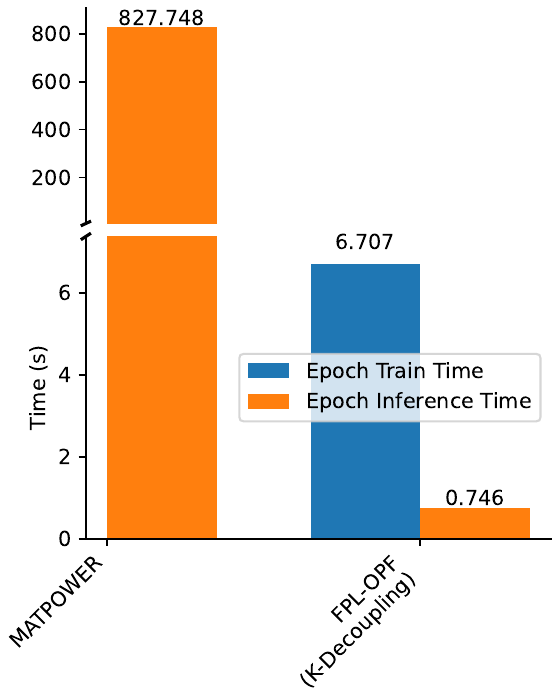}
		\caption{Time-efficiency comparison on PGLib 500-bus system. Blue bars report training time per epoch, and orange bars report inference time on the test set.}    
		\label{fig:appendix_time_comparision_500}
	\end{figure*}
	
	\section{Impact of Iteration Budgets on Exact Implicit Differentiation}
	\label{sec:appendix_ablation_deeplde_cap}
	
	\begin{figure*}[ht]    
		\centering
		\begin{subfigure}[t]{0.22 \textwidth}
			\centering
			\includegraphics[width= \textwidth]{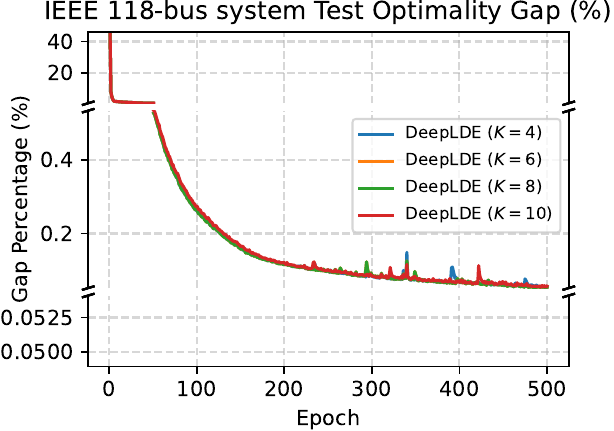}
			\caption{\footnotesize Test optimality gap at IEEE 118-bus system.}
			
			\label{fig:appedix_diff_deeplde_opt_gap_case118}
		\end{subfigure}
				\begin{subfigure}[t]{0.22 \textwidth}
			\centering
			\includegraphics[width= \textwidth]{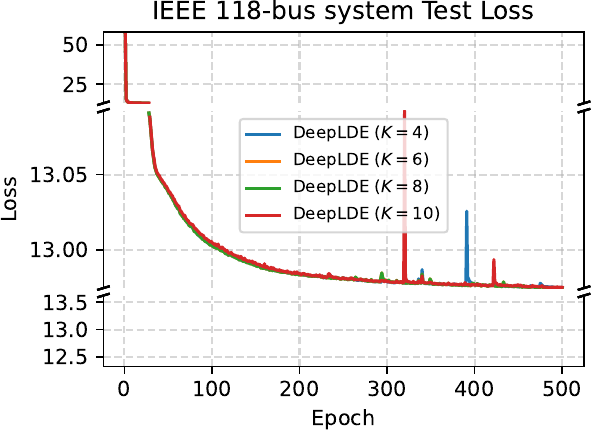}
			\caption{\footnotesize Test loss curve at IEEE 118-bus system.}
			
			\label{fig:appedix_diff_deeplde_loss_curve_case118}
		\end{subfigure}
		\begin{subfigure}[t]{0.22 \textwidth}
			\centering
			\includegraphics[width= \textwidth]{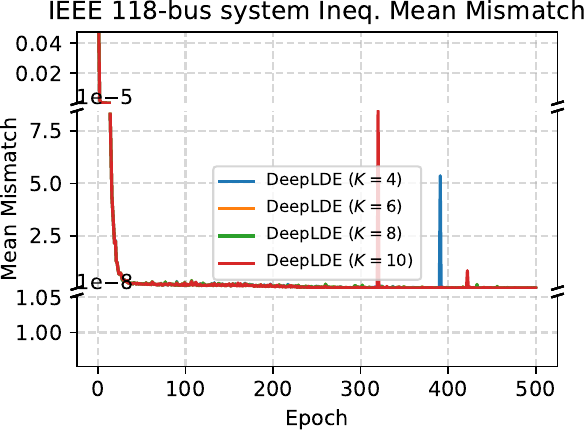}
			\caption{\footnotesize Test Ineq. mean mismatch at IEEE 118-bus system.}
			\label{fig:appedix_diff_deeplde_ineq_mean_mismatch_case118}
		\end{subfigure}
		\begin{subfigure}[t]{0.20 \textwidth}
			\centering
			\includegraphics[width= \textwidth]{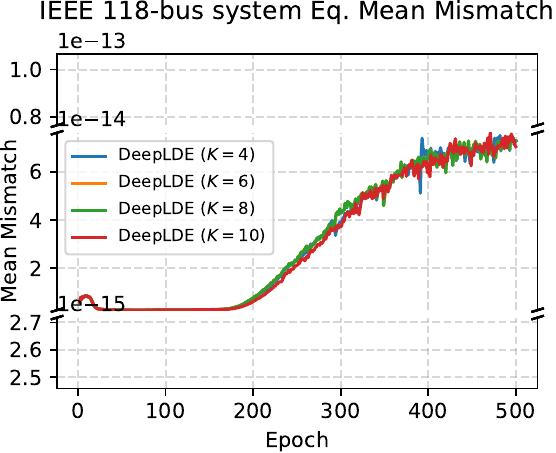}
			\caption{\footnotesize Test Eq. mean mismatch at IEEE 118-bus system.}
			\label{fig:appedix_diff_deeplde_eq_mean_mismatch_case118}
			
		\end{subfigure}
		\caption{Training convergence and constraint mismatch of DeepLDE under varying iteration budgets $K$ on the IEEE 118-bus system. }    
		\label{fig:appendix_diff_deeplde}
	\end{figure*}

		\begin{figure*}[ht]    
		\centering
		\includegraphics[width= 0.35\textwidth]{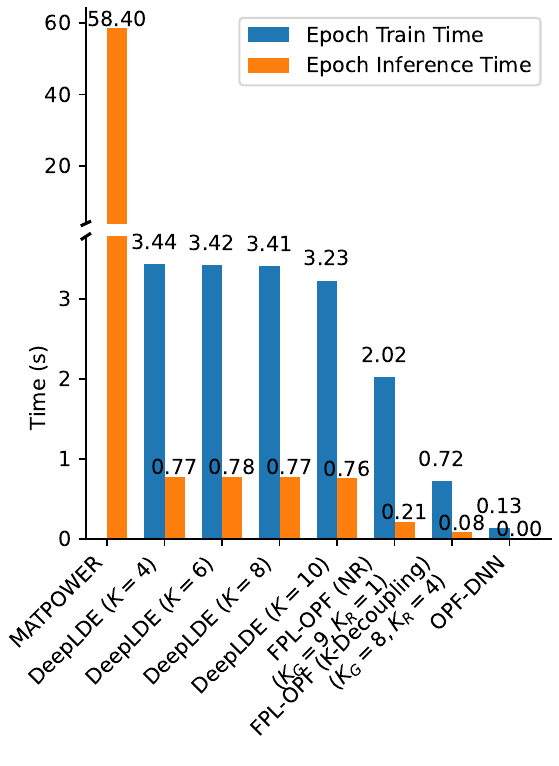}
		\caption{Computational time efficiency of DeepLDE under varying iteration budgets on the IEEE 118-bus system. Blue bars report training time per epoch, and orange bars report inference time on the test set.}    
		\label{fig:appendix_time_diff_deeplde}
	\end{figure*}
	To justify the power flow solver configurations detailed in Table~\ref{tab:detailed_pf_solver_setting}, we conduct an ablation study investigating the impact of the forward-pass iteration budget on the exact implicit differentiation baseline, DeepLDE. All experimental configurations and neural network hyperparameters strictly follow the setup described in Section~\ref{subsec:exp_setup}. The sole parameter varied is the maximum number of iterations allocated to the solver, tested across $K \in \{4, 6, 8, 10\}$ on the IEEE 118-bus system.
	\par
	Figure~\ref{fig:appendix_diff_deeplde} illustrates the training dynamics under these varying iteration budgets, tracking the optimality gap, test loss, and physical constraint mismatches. The empirical results demonstrate that a larger iteration budget fundamentally enhances convergence stability. 
	Most notably, as depicted in Figure~\ref{fig:appedix_diff_deeplde_eq_mean_mismatch_case118}, extending the iteration cap to $K=10$ yields a consistently better and more stable equality mean mismatch compared to lower thresholds, such as $K=4$. 
	Furthermore, Figure~\ref{fig:appendix_time_diff_deeplde} evaluates the computational overhead associated with increasing this budget.
	The comparison reveals that scaling the iteration cap from 4 to 10 does not increase training and inference times per epoch.
	Given the marked improvements in numerical stability and physical feasibility, an iteration budget of 10 is a robust configuration for exact implicit differentiation methods in our evaluations.
		\section{Supplementary Experimental Results for Section~\ref{subsec:results_analysis}}
		\label{sec:appendix_exp_results}
		
		To further evaluate the stability and convergence behavior of the proposed framework, we analyze key performance metrics over the course of training. Figures \ref{fig:appendix_eq_mean_mismatch}--\ref{fig:appendix_obj_cost} illustrate the trajectories of equality mismatches, inequality mismatches, violation numbers, test losses, and objective costs on the IEEE 57-bus, PEGASE 89-bus, IEEE 118-bus, and NESTA 189-bus systems.
		\par
		\subsection{Constraint Satisfaction} Figures \ref{fig:appendix_eq_mean_mismatch}, \ref{fig:appendix_ineq_mean_mismatch}, and \ref{fig:appendix_ineq_viol_num}  depict the equality mean mismatch, inequality mean mismatch, and the number of inequality violations per epoch, respectively. 
		The results highlight a remarkable difference in convergence behavior.
		First, the OPF-DNN method fails to converge to the physical region, exhibiting mismatches orders of magnitude higher than other approaches. This confirms the limitations of the soft penalty method in satisfying strict physical constraints.
		The physics-aware methods, i.e., DC3, DeepLDE, and our FPL-OPF variants, exhibit a rapid and monotonic decrease in inequality constraint violations, effectively reaching the feasible region within the early stages of training.
		\par
		Notably, the FPL-OPF (K-Decoupling) method matches the convergence speed and stability of the exact implicit differentiation method, e.g., DeepLDE, despite using an approximate implicit gradient.
%		As observed in Figure \ref{fig:appedix_case118_ineq_mean_mismatch}, the mismatch for the FPL-OPF framework rapidly descends to the $10^{-7}$ to $10^{-8}$ range.
		As observed across the benchmarks, particularly in more challenging NESTA 189-bus system in Figures \ref{fig:appedix_case189_eq_mean_mismatch} and \ref{fig:appedix_case189_ineq_mean_mismatch}, the mismatch for the FPL-OPF framework rapidly descends to a highly accurate range.
		This empirical evidence across increasingly complex power grids demonstrates that the gradients derived from the K-Decoupling approximation are consistently robust and accurate enough to drive the neural network parameters to a strictly feasible region, matching the fidelity of computationally expensive exact implicit gradients.
		\subsection{Optimality and Loss Convergence} 
		The training dynamics of the test loss and objective cost are presented in Figures \ref{fig:appendix_loss_curve} and \ref{fig:appendix_obj_cost}.
		On smaller systems, the FPL-OPF framework demonstrates smooth convergence trajectories comparable to DeepLDE.
		However, on the large-scale NESTA 189-bus system,  FPL-OPF exhibits significantly more stable and superior convergence behavior shown in Figures \ref{fig:appedix_case189_loss_curve} and \ref{fig:appendix_case189_obj_cost} compared to DeepLDE, which struggles with larger optimality gaps. While OPF-DNN appears to achieve lower objective costs, this is an artifact of its failure to satisfy physical constraints. Conversely, FPL-OPF consistently minimizes the objective function while strictly adhering to the feasible region, validating the robustness and scalability of the proposed fixed-point layer in balancing optimality and physical constraints even for highly complex power grids.

	\end{document}